\documentclass[conference,compsoc]{IEEEtran}

\usepackage{times}
\usepackage{epsfig}
\usepackage{graphicx}
\usepackage{amsthm}
\usepackage{amssymb}
\usepackage{subfigure}
\usepackage{overpic}
\usepackage{booktabs}
\usepackage[utf8]{inputenc}
\usepackage{algorithm}
\usepackage{algorithmicx}
\usepackage[noend]{algpseudocode}
\usepackage{amsmath}
\usepackage{setspace}

\usepackage[numbers,compress]{natbib}

\usepackage{enumerate}
\usepackage{enumitem}
\setitemize{itemsep=0pt,partopsep=0pt,parsep=\parskip,topsep=5pt}
\usepackage[pagebackref=false,breaklinks=true,colorlinks,bookmarks=true]{hyperref}

\floatname{algorithm}{Algorithm}
\algnewcommand{\IIf}[1]{\State\algorithmicif\ #1\ \algorithmicthen}
\algnewcommand{\EndIIf}{\unskip\ \algorithmicend\ \algorithmicif}

\theoremstyle{plain}
\newtheorem{theorem}{Theorem}
\newtheorem{lemma}{Lemma}
\newtheorem{claim}{Claim}
\newtheorem{remark}{Remark}
\newtheorem{corollary}[lemma]{Corollary}
\newtheorem{definition}{Definition}
\theoremstyle{remark}
\newtheorem{case}{-- Case}

\graphicspath{{figures/}}

\begin{document}

\title{Rational Uniform Consensus with General Omission Failures}

\author{\IEEEauthorblockN{Yansong Zhang}
\IEEEauthorblockA{Beijing Jiaotong University \\
20120169@bjtu.edu.cn}
\and
\IEEEauthorblockN{Bo Shen}
\IEEEauthorblockA{Beijing Jiaotong University \\
bshen@bjtu.edu.cn}
\and
\IEEEauthorblockN{Yingsi Zhao}
\IEEEauthorblockA{Beijing Jiaotong University \\
yszhao@bjtu.edu.cn}
}

\maketitle

\begin{abstract}
Generally, system failures, such as crash failures, Byzantine failures and so on, are considered as common reasons for the inconsistencies of distributed consensus and have been extensively studied. In fact, strategic manipulations by rational agents do not be ignored for reaching consensus in distributed system. In this paper, we extend the game-theoretic analysis of consensus and design an algorithm of rational uniform consensus with general omission failures under the assumption that processes are controlled by rational agents and prefer consensus. Different from crashing one, agent with omission failures may crash, or omit to send or receive messages when it should, which leads to difficulty of detecting faulty agents. By combining the possible failures of agents at the both ends of a link, we convert omission failure model into link state model to make faulty detection possible. Through analyzing message passing mechanism in the distributed system with $n$ agents, among which $t$ agents may commit omission failures, we provide the upper bound on message passing time for reaching consensus on a state among nonfaulty agents, and message chain mechanism for validating messages. And then we prove our rational uniform consensus is a Nash equilibrium when $n>2t+1$, and failure patterns and initial preferences are $blind$ (an assumption of randomness). Thus agents could have no motivation to deviate the consensus. Our research strengthens the reliability of consensus with omission failures from the perspective of game theory.
\end{abstract}

\begin{IEEEkeywords}
Consensus; Game Theory; Omission Failures; Message Passing; Distributed Computing
\end{IEEEkeywords}

\section{Introduction}\label{sec:Introduction}
How to reach consensus despite failures is a fundamental problem in distributed computing. In consensus, each process proposes an initial value and then executes a unique consensus algorithm independently. Eventually all processes need to agree on a same decision chosen from the set of initial values even if there may be some system failures, such as crash failures, omission failures and Byzantine failures \cite{1990book}. In crash model, processes can get into failure state by stopping executing the remaining protocol. In omission model, processes can get into failure state by omitting to send or receive messages. And in Byzantine model, processes can fail by exhibiting arbitrary behavior. Extensive studies have been conducted on fault-tolerant consensus. 

Moreover, two kinds of consensus problems are usually distinguished. One is non-uniform version (usually called ``consensus" directly) where no two nonfaulty processes decide differently. The other is uniform version (called ``uniform consensus") where no two processes (whether correct or not) decide on different values. We believe that consensus protocols cannot simply replace uniform consensus protocols, because the condition of non-uniform consensus is inadequate for many applications \cite{1994report}. From \cite{2004JOA}, uniform consensus is harder than consensus because one additional round is needed to decide. And uniform consensus is meaningless with byzantine failures.

Recently, there is an increasing interest on distributed game theory, in which processes are selfish called rational agents. Combining distributed computing with algorithmic game theory is an interesting research area enriching the theory of fault-tolerant distributed computing. In this framework, agents may deviate from protocols in order to increase their own profits according to utility functions. In \cite{2012arxiv}, this kind of deviation is referred to as $strategic$ $manipulation$ of distributed protocol. This research is necessary in some practical scenarios, in which each process has selfish incentives. Clearly, the goal of distributed computing in the context of game theory is to design algorithms for reaching Nash equilibrium, in which all agents have no incentive to deviate from the algorithms. Perhaps, this framework has been investigated and formalized for the first time in the context of secret sharing and multiparty computation \cite{2004STOC,2006CRYPTO,2010TCC,2011PODC}. More recently, some fundamental tasks in distributed computing such as leader election and consensus, have been studied from the perspective of game theory \cite{2018PODC,2019ACMTrans,2021CRYPTO,2016PODC,2017IPDPS,2020AAMAS,2020OPODIS,2021Blockchain,2021DISC}.

Following this new line of research, we combine fault-tolerant consensus with rational agents and study the rational uniform consensus problem in synchronous round-based system, where every agent has its own preference on consensus decisions. Thus an algorithm of rational uniform consensus needs to be constructed. And for each agent, its utility is not less with following the consensus algorithm than with deviating from the algorithm. That achieves a Nash equilibrium. It is easy to see that standard consensus algorithms cannot reach equilibrium and they can be easily manipulated by even a single rational agent. Several researches on rational consensus have been conducted \cite{2012arxiv,2014PODC,2016PODC,2017IPDPS,2012ICALP,2020AAMAS,2020OPODIS,2021Blockchain}, but none of them consider the uniform property. And most studies on rational consensus only support that there are crash failures or no system failures. We argue that omission failures, which are more subtle and complicated than crashing one, cannot be ignored for reaching uniform consensus. In this paper, we pay attention to a distributed system with $n$ agents, among which $t$ agents may experience omission failures. In this setting, we extend the game-theoretic analysis of consensus. Specifically, our contributions in this paper include:
\begin{itemize}
    \item We utilize a punishment mechanism to convert omission failure model into link state model, which makes faulty detection more direct. In the link state model, faulty links never recover whether or not omission failures recover. Therefore, it can provide an idea to simplify the problem of faulty recovery in distributed computing.
    \item An almost complete mechanism analysis is given for message passing in the distributed system with general omission failures. Then, we provide the upper bound $x+1$ on message passing time for reaching consensus on a link state. The upper bound determines the round complexity of our algorithm. Next, a message chain mechanism is introduced for validating messages.
    \item An algorithm of rational uniform consensus with agent omission failures is presented for any $n>2t+1$. We give a complete formal proof of correctness of our algorithm. The proof shows that our consensus is a Nash equilibrium.
\end{itemize}

\subsection{Related Work}
From the view point of modeling methods about agents, the research framework for distributed game theory in the literature may be divided into three categories. In the first category, all of the agents in distributed system are controlled by rational agents preferring consensus and some of them may randomly fail by system failures. Bei et al. \cite{2012arxiv} study distributed consensus tolerating both unexpected crash failures and strategic manipulations by rational agents. They consider agents may fail by crashing. However, the correctness of their protocols needs a strong requirement that it must achieve agreement even if agents deviate. Afek et al. \cite{2014PODC} propose two basic rational building blocks for distributed system and present several fundamental distributed algorithms by using these building blocks. However, their protocol is not robust against even crash failures. Halpern and Vilaca \cite{2016PODC} present a rational fair consensus with rational agents and crash failures. They use failure pattern to describe the random crash failures of agents. Clementi et al. \cite{2017IPDPS} study the problem of rational consensus with crash failures in the synchronous gossip communication model. The protocols of Halpern et al. and Clementi et al. do not tolerate omission failure but we think the consideration to it is necessary. Harel et al. \cite{2020OPODIS} study the equilibria of consensus resilient to coalitions of $n-1$ and $n-2$ agents. They give a separation between binary and multi valued consensus. However, they assume that there are no faulty agents. 

The second category is named rational adversary. Groce et al. \cite{2012ICALP} study the problem of Byzantine agreement with a rational adversary. Rather than the first model, they assume that there are two kinds of processes: one is honest and follows the protocol without question; the other is a rational adversary and prefers disagreement. Amoussou-Guenou et al. \cite{2020AAMAS} study byzantine fault-tolerant consensus from the game theory point. They model processes as rational players or byzantine players and consensus as a committee coordination game. In \cite{2020AAMAS}, the byzantine players have utility functions and strategies, which can be regarded as rational adversaries similar to \cite{2012ICALP}. In our opinion, this framework limits the scope of the byzantine problem.  

Finally, the BAR framework (Byzantine, Altruistic, and Rational) proposed in \cite{2005SOSP}. In \cite{2021Blockchain}, Ranchal-Pedrosa and Gramoli study the gap between rational agreements that are robust against byzantine failures and rational agreements that are robust against crash failures. Their model consists of four different types of players: correct, rational, crash or byzantine, which is similar to the BAR model. They consider that rational players prefer to cause a disagreement than to satisfy agreement, which we view as a bit limited because only referring rational players as rational adversaries that is one of the byzantine problems. Moreover, no protocols are proposed in \cite{2021Blockchain}.

\subsection{Road Map}
The rest of the paper is organized as follows. Section 2 describes the model that we are working in. Section 3 presents the algorithm of rational uniform consensus for achieving Nash equilibrium and proves it correct. Section 4 concludes the paper.

\section{Model}\label{sec:Model}
We consider a synchronous system with $n$ agents and each of agent has a unique and commonly-known identify in $N=\{ 1,...,n \}$. Execution time is divided into a sequence of rounds. Each round is identified by the consecutive integer starting from 1. There are three successive phases in a round: a $send \ phase$ in which each agent sends messages to other agents in system, a $receive \ phase$ in which each agent receives messages that are sent by other agents in the send phase of the same round, and a $computation \ phase$ where each agent verifies and updates the value of local variables and executes local computation based on the messages sent and received in that round. We assume that every pair of agent $i$ and $j$ in $N$ is connected by a reliable communication link denoted $link_{ij}$. For an agent $i$, all links in the system can be divided into two types: direct link $link_{ij}$ where $j\in N$, and indirect link $link_{kp}$ where neither $k$ nor $p$ is equal to $i$.

\subsection{Failure Model}
Here the $general$ $omission$ $failures$ \cite{2005ICTAC}, which occur in agents and not in communication links \cite{2020phd}, are considered. That is an agent crashes or experiences either send omissions or receive omissions. And send omission means that the agent omits sending messages that it is supposed to send. Receive omission means that the agent omits receiving messages that it should receive. We define that agent omission failures never recover. We argue that our protocol also works even if failures could recover, but proving this seems more complicated. It is easy to see that crash failure can be converted to omission failure because if an agent crashes, it must omit to send and receive messages with all other agents after it has crashed. We assume that there are $t$ agents undergoing general omission failures. 

Based on the failure model, we divide the agents in the system into three types:
\begin{itemize}
    \item \textit{good agent}. Good agents do not have omission failures.
    \item \textit{risk agent}. Risk agents experience omission failures but we temporarily consider them as correct agents in our protocol. 
    \item \textit{faulty agent}. Faulty agents have omission failures with more than $t$ agents.
\end{itemize}
It is easy to see that $t$ is the sum of the number of risk and faulty agents. We treat good agents and risk agents as $nonfaulty \ agents$ uniformly. Since send omission and receive omission are symmetrical, for example, $i$ omits to send messages with $j$ and $j$ omits to receive messages with $i$ have the same phenomenon, we may not be able to directly detect the states of some agents with omission failures. Thus we call them risk agents and consider them as correct agents. For an agent that has omission failures with more than $t$ agents, it must have omission failures with at least one good agent and then clearly we can know it is a faulty agent. 

Due to the symmetry of agent omission failures, we model the agent omission failures as the link state problem by a punishment mechanism. Specifically, in our protocol, if an agent $i$ receives no messages from $j$ in a round, then in the following rounds, $i$ sends no messages to $j$ and does not receive messages from $j$ \cite{2004SPAA}. Thus whether send omission or receive omission will cause the link interruption. So in a round, we divide each link $link_{ij}$ into three types: $correct \ link$, where neither $i$ nor $j$ experiences omission failures in this round, $faulty \ link$, where at least one of $i$ and $j$ has omission failures with the other one in the round, and $unknown\mbox{-}state \ link$, where the state of $link_{ij}$ in this round is unknown to another agent $k$. It is easy to see that we can determine the type of an agent by the number of correct direct links of it, which is the fault detection method in our protocol. Similarly, faulty links never recover under this punishment mechanism whether or not omission failures recover.

\subsection{Consensus}
In the consensus problem, we assume that every agent $i$ has an initial preference $v_i$ in a fixed value set $V$. (We follow the concept of initial preference in \cite{2016PODC}.) We are interested in Uniform Consensus in this paper. A protocol solving uniform consensus must satisfy the following properties \cite{2004JOA}:
\begin{itemize}
    \item \textbf{Termination}: Every correct agent eventually decides.
    \item \textbf{Validity}: If an agent decides $v$, then $v$ was the initial value of some agent.
    \item \textbf{Uniform Agreement}: No two agents (whether correct or not) decide on different values.
\end{itemize}
To solve uniform consensus in presence of agent omission failures, we assume that $n>2t+1$ and $n\geq 3$.

In uniform consensus, an agent's final decision must be one of the following formalized types:
\begin{itemize}
    \item $\perp$: It means that there is no consensus. $\perp$ is a punishment for inconsistency.
    \item $\parallel$: It means no decision. Deciding $\parallel$ is not ambiguous with validity, as $\parallel$ cannot be proposed \cite{2004SPAA}. It does not affect the final consensus outcome.
    \item $v\in V$: It satisfies the property of validity, which must be the initial preference of some agent.
\end{itemize}

\subsection{Rational agent}
We consider that distributed processes act as rational agents according to the definition in game theory. Each agent $i$ has a utility function $u_i$. We assume that agents have solution preference \cite{2014PODC} and an agent's utility depends only on the consensus value achieved. Thus for each agent $i$, there are three values of $u_i$ based on the consensus value achieved: (\romannumeral1) $\beta_0$ is $i$'s utility if $i$'s initial preference $v_i$ is decided; (\romannumeral2) $\beta_1$ is $i$'s utility if there is a consensus value which is not equal to $i$'s initial preference; (\romannumeral3) $\beta_2$ is $i$'s utility if there is no consensus. It is easy to see that $\beta_0 > \beta_1 > \beta_2$, and our results can easily be extended to deal with independent utility function for each agent.

The $strategy$ of an agent $i$ is a local protocol $\sigma_i$ satisfying the system constrains. $i$ takes actions according to the protocol $\sigma_i$ in each round. That is $\sigma_i$ is a function from the set of messages received to actions. Each agent chooses the protocol in order to maximize its expected utility. Thus there are $n$ local protocols chosen by every agent, which is called $strategy \ profile \ \Vec{\sigma}$ in game theory. The equilibrium is a $strategy \ profile$, where each agent cannot increase its utility by deviating if the other agents fix their strategies. For each agent $i$, if the local protocol $\sigma_i$ is our consensus algorithm when reaching an equilibrium, then we say that $rational \ consensus$ is a $Nash \ equilibrium$. Formally, if a strategy profile (or consensus) $\Vec{\sigma}$ is a $Nash \ equilibrium$, then for all agents $i$ and all strategies $\sigma_i\overset{'}{}$ for $i$, it must have $u_i(\sigma_i\overset{'}{}, \Vec{\sigma}_{\mbox{-}i})\leq u_i(\Vec{\sigma})$.

\section{Rational Uniform Consensus with General Omission Failures}\label{sec:Protocol design}
\subsection{An Rational Uniform Consensus in synchronous systems}
In order to reach rational uniform consensus that can tolerate omission failures, our protocol adopts a simple idea from an early consensus protocol \cite{2004SPAA}: An agent does not send or receive any messages to those agents that did not send messages to it previously. Then we convert the omission failure model which cannot be detected into the link state model which can be detected by agents in each round. However, the presence of rational agents makes protocol more complicated. It requires protocol could prevent the manipulation of rational agents. Hence, the security of the algorithm needs to be improved from three aspects. The first is interacting the latest network link states and message sources in each round. The update process of the latest link states within each agent depends on complete message chains and we can obtain a unified decision round and decision set from message passing mechanism in omission failure environment. The second is using secret sharing for agents' initial preferences \cite{1979secret}. It encrypts the initial preferences so as to prevent an agent knowing the values of other agents in advance. The third is signing each message with a random number and marking faulty links by faulty random numbers \cite{2012arxiv}. This can improve the difficulty of a rational agent to do evil.

The protocol is described in Algorithm \ref{algorithm 1}. In more detail, we proceed as follows. 

Initially, each agent $i$ generates a random number $proposal_i$ which is used for consensus election later (line 2). Then $i$ computes two random 1-degree polynomials $q_i$ and $b_i$ with $q_i(0)=v_i$ and $b_i(0)=proposal_i$ respectively (line 3). They satisfy (2,n) threshold which means that an agent $j \neq i$ can restore $v_i$ or $proposal_i$ if it knows more than two pieces of $q_i$ or $b_i$. Then $i$ initialize set $lost_i$, $NS_i^0$, $HS_i$, $decision_i$ and set $consensus_i$ (line 4); we discuss these in more detail below. Then $i$ generates the faulty random number $X\mbox{-}random_i^1[k][l_{ij}]$ for each agent $k\neq i$ and each direct link $l_{ij}$ that is the abbreviation of $link_{ij}$ for the link between $i$ and $j$ (lines 5-7; $l_{ij}^r $ represents the $link_{ij}$ in round $r$). And the message random number $random_i^1$ for round 1 is randomly chosen from $\{0,...,n-1\}$ (line 8). For each link, $i$ generates $n-1$ faulty random numbers and then sends them to other agents respectively in round 1. So we can get that $X\mbox{-}random_i^1$ contains $(n-1)^2$ faulty random numbers in total. Then $i$ puts $X\mbox{-}random_i^1$ and $random_i^1$ into \Call{X-Random}{} and \Call{Random}{} respectively (lines 9 and 10), where \Call{X-Random}{} is a function storing all faulty random numbers known to $i$ and \Call{Random}{} stores all message random numbers. Agents can invoke these two functions to verify random numbers. Specifically, inputting the id, link and round to invoke \Call{X-Random}{}, and inputting id and round to invoke \Call{X-Random}{}.

There are $t+4$ rounds in total and each round has three phases. In phase 1 of round $r$, $1 \leq r \leq t+4$, $i$ only sends messages to each agent $j$ who doesn't belong to $lost_i$ that is a set of agents that have omission failures with $i$ detected by $i$ (line 15). If $1 \leq r \leq t+3$, $i$ sends $random_i^r$ and $NS_i^{r-1}$ to $j$. And if $1\leq r\leq t+2$, $i$ also sends $X\mbox{-}random_i^r[j]$ which contains $n-1$ random numbers that each of corresponds to a direct link respectively (lines 16 and 17). If $r=1$, $i$ also sends the piece of $q_i$, $q_i(j)$ and the piece of $b_i$, $b_i(j)$ to $j$ (line 16). If $r=t+3$, $i$ also sends all the secret shares $q_l(i)$ and $b_l(i)$ ($l \neq j$) that it has received from other agents (line 18). It is easy to see that the piece $q_l(i)$ and $b_l(i)$ must be in pairs. That is if $i$ restores $v_j$, then it can also restore $proposal_j$. Finally, if $r=t+4$, $i$ only sends $consensus_i$ to $j$ (line 19). For each agent $i$, $consensus_i$ is the set of all consensus values calculated and received by $i$. Hence if the algorithm is executed validly, $|consensus_i|$ must be equal to 1. 

In phase 2 of round $r$, $1 \leq r \leq t+4$, $i$ only receives messages from agents that are not in set $lost_i$ (line 22). And if there are no messages received from an agent $j$, $j \notin lost_i$, $i$ adds $j$ to $lost_i$ (line 29). Otherwise, $i$ stores the received information (lines 24-28). $\{ NS^{r-1} \}$ and $\{ random^r \}$ are the sets of all new link states $NS_j^{r-1}$ and message random numbers $random_j^r$ respectively received by $i$ from each agents $j \notin lost_i$ in round $r$ (line 26). Correspondingly, the elements in $\{ NS^{r-1} \}$ and $\{ random^r \}$ are one-to-one correspondence. Specially, if $|lost_i|>t$, this means $i$ knows that it becomes a faulty agent and then $i$ must decide $\parallel$ directly and no longer run in later rounds (line 31). And we say that $\parallel$ means agent $i$ does not decide in the end, which has no influence on the solution.

In phase 3 of round $r \leq t+3$, $i$ firstly uses $NS_i^{r-1}$ to update $NS_i^r$ which is useful for the update and verification of link states (line 35). Then $i$ invokes the function \Call{VerifyAndUpdate}{} to verify and update $NS_i^r$ and $HS_i$ by $\{ NS^{r-1} \}$ and $\{ random^r \}$ (line 36; See Algorithm \ref{algorithm 2} for details). $NS_i^r$ is the latest states known to $i$ of all links in the system in round $r$. $HS_i$ is the historical link states including $t+3$ rounds in total. If $r \leq t+2$, $i$ generates the message random number $random_i^{r+1}$ and the faulty random numbers $X\mbox{-}random_i^{r+1}$ for round $r+1$, which will be sent to other agents in round $r+1$, and then puts these random numbers into \Call{X-Random}{} and \Call{Random}{} respectively (lines 37-43). Then if $r=t+3$, $i$ last updates $HS_i$ by $NS_i^{t+3}$ (line 45). Specifically, if a link $l_{kp}$ is faulty in round $m$ in $NS_i^{t+3}$, then changing the state of $l_{kp}$ into fault from round $m$ to round $t+3$ in $HS_i$. This is the last time modifying $HS_i$. And following that, $i$ utilizes $HS_i$ to find the $decision \ round$ $m^*$ from round 1 to round $t+2$, which is the first $reliable \  round$ in $HS_i$ (lines 46-48). We follow the concept of $clean \ round$ in \cite{2016PODC}. The number of faulty agents does not increase in $clean \ round$ and the previous round of $clean \ round$ is $reliable \ round$. Specially, we say that $reliable \ round$ cannot be round 0, so that the first $reliable \ round$ is the previous round of the second $clean \ round$ if the first $clean \ round$ is round 1. In $HS_i^r$, if less than $n-t-1$ links to agent $j$ are correct, then $j$ is a faulty agent. Otherwise, $j$ is a nonfaulty agent. We define that it must remove the explicitly faulty agents when computing the state of $j$ in round $r$ by $HS_i$. Then $i$ computes the $decision \ set \ D$ that is the set of nonfaulty agents in $HS_i^{m\overset{*}{}}$ (line 49). And then $i$ uses all the secret shares it has received in round $t+3$ to try to restore the initial preference and proposal of each agent $j \in D$ (lines 50-53). If $i$ can reconstruct the values of all agents $\in D$, then $i$ must know all proposals of these agents. Then $i$ computes the consensus proposal (lines 54-63). Firstly, $i$ sorts all the proposals in $D$ and finds the set $C$ of agents with the second max proposal value (line 55). Then if there is only one agent in $C$, $i$ puts the initial preference of this agent into $consensus_i$ (line 56). If there are more than one agents in $C$, that is more than one agents have the same second max proposal value $pr$ in $D$, (we say that the probability is extremely low,) then $i$ uses the $pr$ to mod the agent number of $C$ and gets $S$ (lines 60-62). In this case, $i$ finally puts $v_j$ into its $consensus_i$ where $j$ is the ($S+1$)st highest id in $C$ (line 63). Finally if there is no agent in $C$, then the proposals must be the same for all agents in $D$ and the second max proposal value does not exist. Thus $i$ uses the same proposal to mod the agent number of $decision \ set \ D$ and gets $S$ (line 58). Similarly, $i$ elects the initial preference of the agent with the ($S+1$)st highest id in $D$ (line 59). But if $i$ can not restore all the values of agents $\in D$, it does nothing and keeps $consensus_i$ as the empty set. Finally, if $r=t+4$, if $consensus_i$ contains only one value, then $i$ makes a decision (lines 65-67). Otherwise, an inconsistency is detected and $i$ decides $\perp$ (line 69).

\begin{algorithm*}
    \caption{agent $i$'s uniform consensus protocol with initial value $v_i$ $(n>2t+1)$}\label{algorithm 1}
    \begin{algorithmic}[1] %每行显示行号
        % \ensure $(ns^r_i, hs_i)$ \textbf{or} $decided$
        \Function {Consensus}{$v_i$}
            \State $proposal_i\gets$ a random number
            \State $q_i, b_i\gets$ random 1 degree polynomials with $q_i(0)=v_i$ and $b_i(0)=proposal_i$
            \Comment{(2,n) threshold secret sharing}
            \State $lost_i\gets\emptyset;NS_i^0\gets \emptyset; HS_i\gets \emptyset; decision_i\gets none; consensus_i\gets \emptyset$
            \ForAll{$ j \neq i$}
                \ForAll{$k \neq i$}
                    \State $X\mbox{-}random_i^1[k][l_{ij}]\gets$ a random bit
                \EndFor
            \EndFor
            \State $random_i^1\gets$ random value in $\{0,...,n-1\}$
            \State puts $X\mbox{-}random_i^1$ into \Call{X-Random}{}
            \State puts $random_i^1$ into \Call{Random}{}
                
            \State
            \For{round $r = 1 \to t+4$}
                \State $\{ NS^{r-1} \} \gets \emptyset; \{ random^r \} \gets \emptyset$
                \State Phase 1: send phase
                \ForAll{$j \notin lost_i$ \textbf{and} $j \neq i$}
                    \If{$r=1$} Send $\langle q_i(j), b_i(j), NS_i^{r-1}, X\mbox{-}random_i^r[j], random_i^r\rangle$ to j \EndIf
                    \If{$2\leq r \leq t+2$} Send $\langle NS_i^{r-1},X\mbox{-}random_i^r[j],random_i^r\rangle$ to j \EndIf
                    \If{$r=t+3$} Send $\langle NS_i^{r-1},(q_l(i))_{l\neq j},(b_l(i))_{l\neq j}, random_i^r\rangle$ to j \EndIf
                    \If{$r=t+4$} Send $consensus_i$ to j \EndIf
                \EndFor
                    
                \State
                \State Phase 2: receive phase
                \ForAll{$j \notin lost_i$ \textbf{and} $j \neq i$}
                    \If{new message has received from j}
                        \State puts $X\mbox{-}random_j^r[i]$ into \Call{X-Random}{} \Comment{round 1 to t+2}
                        \State puts $random_j^r$ into \Call{Random}{} \Comment{round 1 to t+3}
                        \State $\{ NS^{r-1} \} \gets \{ NS^{r-1} \} \cup NS_j^{r-1};\{ random^r \}\gets \{ random^r \} \cup random_j^r$ \Comment{round 1 to t+3}
                        \State save $(q_l(j))_{l \neq i}$ and $(b_l(j))_{l \neq i}$ \Comment{round t+3}
                        \State$consensus_i\gets consensus_i \cup consensus_j$ \Comment{round t+4}
                    \Else $\ lost_i\gets lost_i\cup {j}$
                    \EndIf
                \EndFor
                \If{$|lost_i|>t$} \Comment{Faulty agent}
                    \State$\textbf{Decide}(\parallel)$ \Comment{No decision}
                \EndIf
                    
                \State
                \State Phase 3: computation phase
                \If{$r \leq t+3$}
                    \State $NS_i^r\gets NS_i^{r-1}$
                    \State $NS_i^r,HS_i \gets$\Call{VerifyAndUpdate}{$r$, $i$, $NS^r_i$, $HS_i$, $\{NS^{r-1}\}$, $\{random^r\}$}
                    \Comment{Punishment if an inconsistency is detected}
                    \If{$r \leq t+2$} 
                        \State $random_i^{r+1}\gets$ random value in $\{0,...,n-1\}$
                        \ForAll{$ j \neq i$}
                            \ForAll{$k \neq i$}
                                \State $X\mbox{-}random_i^{r+1}[k][l_{ij}]\gets$ a random bit
                            \EndFor
                        \EndFor
                         \State puts $X\mbox{-}random_i^{r+1}$ into \Call{X-Random}{}
                        \State puts $random_i^{r+1}$ into \Call{Random}{}
                    \ElsIf{$r = t+3$}
                        \State\Call{LastUpdate}{$HS_i, NS_i^{t+3}$}\Comment{Update the $Msg\mbox{-}X$ of $HS_i$}
                        \For{$m = 1 \to t+2$}
                            \If{$m$ is the first reliable round in $HS_i$}
                                \State $m\overset{*}{}\gets m;break$ \Comment{decision round}
                            \EndIf
                        \EndFor
                        \State $D\gets$ the set of nonfaulty agents in $HS_i^{m\overset{*}{}}$ \Comment{decision set}
                        \ForAll{$j \in D$}
                            \If{the number of $q_j(l) < 2$ \textbf{and} $j \neq i$} $break$ \EndIf
                            \State $q_j,b_j\gets$ restored by $q_j(l)$ and $b_j(l)$ received
                            \State $v_j,proposal_j\gets q_j(0), b_j(0)$
                        \EndFor
    \algstore{algorithm1}
    \end{algorithmic}
\end{algorithm*}    
\begin{algorithm*}
    \begin{algorithmic}[1]
    \algrestore{algorithm1}
        %the second page   
                        \If{all values are known in $D$}
                            \State $C\gets$ the set of agents with the second max proposal in $D$
                            \If{$|C|=1$} $consensus_i\gets consensus_i \cup v_{C.element}$
                            \ElsIf{$|C|=0$}
                                \State $S\gets the \ same \ proposal$ mod $|D|$ \Comment{the proposals are the same for all agents in $D$}
                                \State$consensus_i\gets consensus_i \cup v_j$, where $j$ is the $(S+1)$st highest id in $D$
                            \Else \Comment{$|C|>1$}
                                \State$pr\gets$ the second max proposal in $D$
                                \State$S\gets pr$ mod $|C|$ 
                                \State$consensus_i\gets consensus_i \cup v_j$, where $j$ is the $(S+1)$st highest id in $C$ 
                            \EndIf
                        \EndIf
                    \EndIf
                \Else \Comment{round t+4}
                    \If{$|consensus_i|=1$}
                        \State $decision_i\gets consensus_i.element$
                        \State$\textbf{Decide}(decision_i)$
                    \Else \Comment{Inconsistency}
                        \State$\textbf{Decide}(\perp)$
                    \EndIf
                \EndIf
            \EndFor
        \EndFunction
    \end{algorithmic}
\end{algorithm*}

The detailed implement of the verification and update protocol in phase 3 is given in Algorithm \ref{algorithm 2}.

Basically, for each link $l_{kp}$, $NS_i^r[l_{kp}]$ is a tuple containing two tuples, $t_A$ and $t_B$. The first tuple $t_A$ represents the state of $l_{kp}$, which contains three types: $Msg\mbox{-}R$, $Msg\mbox{-}X$ and $Msg\mbox{-}O$, representing correct link, faulty link and unknown-state link, respectively. The format of type $Msg\mbox{-}R$ is $(m, k, random_p^m)$, where $m$ is the round of the link state, $k$ is the agent reporting the link state, and $random_p^m$ is the message random number sent by $p$ in round $m$. It is easy to see that if $k$ reports the state of $l_{kp}$ is correct ($Msg\mbox{-}R$) in round $m$, then it must know $random_p^m$. The format of type $Msg\mbox{-}X$ is $(X, m, k, X\mbox{-}random_k^m[l_{kp}])$, where $m$ and $k$ are the same as those in $Msg\mbox{-}R$, $X$ is an identifier, and $X\mbox{-}random_k^m[l_{kp}]$ is the sorted set of faulty random numbers on $l_{kp}$ which is generated by $k$ in round $m-1$. Specifically, the set is sorted by the ids of agents from small to large. The format of type $Msg\mbox{-}O$ is $\emptyset$ because the state of $l_{kp}$ is unknown for $i$. The second tuple $t_B$ describes the source of $t_A$ and has the form $(j, m)$, where $j$ is the agent sending the link state $t_A$ to $i$, and $m$ is the round when $j$ sends it to $i$. Specially, for direct link, when $i$ first updates the state in round $r$, $t_B$ is $\phi$ meaning that the message source is $i$ itself. $HS_i^r[l_{kp}]$ denotes the state of link $l_{kp}^r$ known to $i$ and $r$ could range from 1 to $t+3$. It contains at most two different tuples because the state of $l_{kp}$ in round $r$ can only be detected and reported by $k$ and $p$. The form of each tuple is similar to $t_A$ but the round in $t_A$ must be $r$. And the agents in the two tuples must be different and be $k$ and $p$ respectively. Specially, if the types of two tuples are $Msg\mbox{-}R$ and $Msg\mbox{-}X$, then we think the state of $l_{kp}^r$ is faulty. And if $Msg\mbox{-}O$ and $Msg\mbox{-}O$, then $l_{kp}^r$ is a unknown-state link which is regarded as a correct link when computing $decision \ round$ and $decision \ set$ in round $t+3$.

The pseudocode in Algorithm \ref{algorithm 2} is explained in detail as follows.

$i$ initially generates $T$ from $\{ NS^{r-1} \}$, which represents $i$ receives the messages sent by agent $j \in T$ in round $r$ (line 2). And $i$ also computes set $S$ that is equal to $lost_i$ (line 3). 

Firstly, in phase 1, $i$ updates the states of direct links in round $r$. For each agent $j \in T$, $i$ has received the messages from it in round $r$ so that $i$ updates the $t_A$ of $NS_i^r[l_{ij}]$ to Type $Msg\mbox{-}R$ (line 7). And $i$ must be able to obtain the message random number $random_j^r$ from $\{random^r\}$. Then $i$ invokes \Call{AppendHS}{} to append the state $(r, i, random_j^r)$ into $HS_i^r[l_{ij}]$ (line 8). We stipulate \Call{AppendHS}{} must guarantee that the inputting state satisfies the properties of $HS$ which we have discussed above. For example, each link $l_{kp}$ has at most two different tuples in each round, and they come from different agents, $k$ and $p$. If a state violated the properties of $HS$, \Call{AppendHS}{} would decide $\perp$ and terminate the protocol early. Then for each agent $j \in S$, it has omission failures detected by $i$ because $i$ does not receive a message from it. If the type of link $l_{ij}$ is already faulty in $NS_i^r$ inherited $NS_i^{r-1}$, $i$ does nothing because for a link, $NS$ only records the earliest round when the link has failures (lines 10-11). Otherwise, $i$ updates the $t_A$ to Type $Msg\mbox{-}X$ and appends the new state into $HS_i$ (lines 13-14).

Then in phase 2, $i$ utilizes message chain mechanism to verify the correctness of messages $\{ NS^{r-1} \}$ received in receive phase (lines 17-18). 

\noindent(\textbf{Message Chain Mechanism}) For each agent $j$, its message $NS_j^m$ has the following properties:

Suppose $S_j^m$ is the set of agents that disconnected from $j$ in or before round $m$ and $T_j^m$ is the set of agents that are still connected to $j$ in round $m$. Suppose $X(r)$ represents the $Msg\mbox{-}X$ tuple where the round number is equal to $r$.

\begin{claim}\label{claim 1}
    For link $l_{kp}$ in $NS_j^m$, where $k = j$ and $p \in S_j^m \cup T_j^m$, its state in round $m$ must be known and the number of correct links in $\{ l_{jp} \}$ is greater than or equal to $n-t-1$.
\end{claim}  

\begin{claim} \label{claim 2}
    For link $l_{kp}$ in $NS_j^m$, where $k = j$ and $p \in S_j^m \cup T_j^m$, its state in round $m+1$ and later must be unknown.
\end{claim}

\begin{claim}
    For link $l_{kp}$ in $NS_j^m$, where $k \in S_j^m$ and $p \in S_j^m$, its state in round $m-1$ and later must be unknown.
\end{claim}

\begin{claim}
    For link $l_{kp}$ in $NS_j^m$, where $k \in S_j^m$ and $p \in S_j^m$, if the state of $NS_j^m[l_{jk}]$ is $X(m_1)$ and the state of $NS_j^m[l_{jp}]$ is $X(m_2)$, suppose $m\geq m_1\geq m_2$, then the state of $l_{kp}^{m_1-2}$ must be known.
\end{claim}

\begin{claim}
    For link $l_{kp}$ in $NS_j^m$, where $k \in T_j^m$ and $p \in S_j^m \cup T_j^m$, its state in round $m-1$ must be known and, in round $m$ and later must be unknown.
\end{claim}

\begin{claim}
    For link $l_{kp}$ in $NS_j^m$, where $k \in T_j^m$ and $p \in S_j^m$, if the state of $l_{kp}^{m-1}$ is $Msg\mbox{-}R$, then the state of link $l_{pt}$ in round $m-2$, where $t \in S_j^m$, must be known and its state in round $m-1$ and later must be unknown.
\end{claim}

\begin{claim}
    For link $l_{kp}$ in $NS_j^m$, where $k \in T_j^m$ and $p \in S_j^m$, if the state of $l_{kp}^{m-1}$ is equal to $X(m\overset{'}{})$, where $m\overset{'}{}\leq m-1$, then the state of $l_{pt}^{m\overset{'}{}-2}$, $t \in S_j^m$, must be known.
\end{claim}

Explicitly, we say that the function \Call{VerifyMsgChain}{} is to verify whether a message $NS_j^{r-1} \in \{ NS^{r-1} \}$ violates the above claims. If all not, then continuing to the phase 3. Otherwise, it decides $\perp$ and terminates the protocol early.

Finally, in phase 3, $i$ updates $NS_i^r$ and $HS_i$ by the states in $\{ NS^{r-1} \}$. For a link, $i$ compares its state in $NS_i^r$ with the state in $NS_j^{r-1} \in \{ NS^{r-1} \}$, so as to implement update according to different cases. 

\begin{claim}\label{first round number verification}
    For each link $l_{kp}$ ($k,p\in N$), the agent of $t_A$ must be $k$ or $p$ in all $NS[l_{kp}]$ and $HS[l_{kp}]$.
\end{claim}

\begin{case}\label{case a}
    For the direct link $l_{ij}$ of $i$, if the $t_A$ of $NS_i^r[l_{ij}]$ is $(r,i,random)$ and the $t_A$ of $NS_j^{r-1}[l_{ij}]$ is $(r\overset{'}{},k,random\overset{'}{})$, then $i$ only needs to append the new state into $HS_i$ (lines 39-40).
\end{case}

\begin{claim}
    In Case \ref{case a}, there must be $r\overset{'}{} < r$.
\end{claim}

\begin{case}\label{case b}
    For the direct link $l_{ij}$ of $i$, if the $t_A$ of $NS_i^r[l_{ij}]$ is $(r,i,random)$ and the $t_A$ of $NS_j^{r-1}[l_{ij}]$ is $(X,r\overset{'}{},k,X\mbox{-}random_k^{r\overset{'}{}}[l_{ij}])$, then $i$ detects an inconsistency and decides $\perp$ (lines 41-42).
\end{case}

\begin{case}\label{case c}
    For the direct link $l_{ij}$ of $i$, if the $t_A$ of $NS_i^r[l_{ij}]$ is $(X,r\overset{'}{},k,X\mbox{-}random_k^{r\overset{'}{}}[l_{ij}])$ and the $t_A$ of $NS_j^{r-1}[l_{ij}]$ is $(r\overset{''}{},p,random)$, then $i$ only needs to append the new state into $HS_i$ (lines 43-44).
\end{case}

\begin{claim}
    In Case \ref{case c}, there must be $r\overset{''}{} \leq r\overset{'}{}$.
\end{claim}

\begin{case}\label{case d}
    For the direct link $l_{ij}$ of $i$, if the $t_A$ of $NS_i^r[l_{ij}]$ is $(X,r\overset{'}{},i,X\mbox{-}random_i^{r\overset{'}{}}[l_{ij}])$ and the $t_A$ of $NS_j^{r-1}[l_{ij}]$ is $(X,r\overset{''}{},k,X\mbox{-}random_k^{r\overset{''}{}}[l_{ij}])$, then $i$ only needs to append the new state into $HS_i$ when $r\overset{''}{}=r\overset{'}{}$ or $r\overset{''}{}=r\overset{'}{}+1$, and $i$ must update $NS_i^r$ and $HS_i$ when $r\overset{''}{}=r\overset{'}{}-1$ (lines 45-51). When updating $NS_i^r$, the $t_B$ of $NS_i^r[l_{ij}]$ must be $(j,r)$ because the new state is obtained from $NS_j^{r-1}$ and updated in round $r$.
\end{case}

\begin{claim}
    In Case \ref{case d}, if $k=i$, the $t_A$ of $NS_j^{r-1}[l_{ij}]$ must be the same as the $t_A$ of $NS_i^r[l_{ij}]$, and if $k=j$, it must have $0\leq |r\overset{'}{}-r\overset{''}{}| \leq 1$.
\end{claim}

\begin{case}\label{case e}
    For the direct link $l_{ij}$ of $i$, if the $t_A$ of $NS_i^r[l_{ij}]$ is $(X,r\overset{'}{},j,X\mbox{-}random_j^{r\overset{'}{}}[l_{ij}])$ and the $t_A$ of $NS_j^{r-1}[l_{ij}]$ is $(X,r\overset{''}{},k,X\mbox{-}random_k^{r\overset{''}{}}[l_{ij}])$, then $i$ does nothing.
\end{case}

\begin{claim}
    In Case \ref{case e}, if $k=j$, the $t_A$ of $NS_j^{r-1}[l_{ij}]$ must be the same as the $t_A$ of $NS_i^r[l_{ij}]$, and if $k=i$, it must have $r\overset{''}{}=r\overset{'}{}+1$.
\end{claim}

\begin{case}\label{case f}
    For the indirect link $l_{kp}$ of $i$, if the $t_A$ of $NS_i^r[l_{kp}]$ is $(r\overset{'}{},y,random)$ and the $t_A$ of $NS_j^{r-1}[l_{kp}]$ is $(r\overset{''}{},z,random\overset{'}{})$, then $i$ only needs to append the new state into $HS_i$ when $r\overset{''}{} \leq r\overset{'}{}$, and $i$ must update $NS_i^r$ and $HS_i$ when $r\overset{''}{} > r\overset{'}{}$ (lines 53-58).
\end{case}

\begin{case}\label{case g}
    For the indirect link $l_{kp}$ of $i$, if the $t_A$ of $NS_i^r[l_{kp}]$ is $(r\overset{'}{},y,random)$ and the $t_A$ of $NS_j^{r-1}[l_{kp}]$ is $(X,r\overset{''}{},z,X\mbox{-}random_z^{r\overset{''}{}}[l_{kp}])$, then $i$ needs to update $NS_i^r$ and append the new state into $HS_i$. (lines 59-60).
\end{case}

\begin{claim}
    In Case \ref{case g}, if $z=y$, it must have $r\overset{'}{} < r\overset{''}{}$, and if $z\neq y$, it must have $r\overset{'}{} \leq r\overset{''}{}$.
\end{claim}

\begin{case}\label{case h}
    For the indirect link $l_{kp}$ of $i$, if the $t_A$ of $NS_i^r[l_{kp}]$ is $(X,r\overset{'}{},y,X\mbox{-}random_y^{r\overset{'}{}}[l_{kp}]))$ and the $t_A$ of $NS_j^{r-1}[l_{kp}]$ is $(r\overset{''}{},z,random)$, then $i$ only needs to append the new state into $HS_i$. (lines 62-63).
\end{case}

\begin{claim}
    In Case \ref{case h}, if $z=y$, it must have $r\overset{'}{} > r\overset{''}{}$, and if $z\neq y$, it must have $r\overset{'}{} \geq r\overset{''}{}$.
\end{claim}

\begin{case}\label{case i}
    For the indirect link $l_{kp}$ of $i$, if the $t_A$ of $NS_i^r[l_{kp}]$ is $(X,r\overset{'}{},y,X\mbox{-}random_y^{r\overset{'}{}}[l_{kp}]))$ and the $t_A$ of $NS_j^{r-1}[l_{kp}]$ is $(X,r\overset{''}{},z,X\mbox{-}random_z^{r\overset{''}{}}[l_{kp}])$, then $i$ does nothing when $z=y$, and $i$ appends the new state into $HS_i$ when $z \neq y$. Specially, $i$ also updates $NS_i^r$ using the new state received if $r\overset{'}{} > r\overset{''}{}$. (lines 64-70).
\end{case}

\begin{claim}\label{last round number verification}
    In Case \ref{case i}, if $z=y$, the $t_A$ of $NS_j^{r-1}[l_{kp}]$ must be the same as the $t_A$ of $NS_i^r[l_{kp}]$, and if $z\neq y$, it must have $0 \leq |r\overset{'}{} - r\overset{''}{}| \leq 1$.
\end{claim}

\begin{case}\label{case j}
    If the $t_A$ of $NS_j^{r-1}[l_{kp}]$ is $\emptyset$, then $i$ does nothing (lines 28-29).
\end{case}

\begin{case}\label{case k}
    For the indirect link $l_{kp}$ of $i$, if the $t_A$ of $NS_i^r[l_{kp}]$ is $\emptyset$ and the $t_A$ of $NS_j^{r-1}[l_{kp}]$ is $(r\overset{'}{},z,random)$ or $(X,r\overset{'}{},z,X\mbox{-}random_z^{r\overset{'}{}}[l_{kp}]))$, then $i$ needs to update $NS_i^r$ and append the new state into $HS_i$. (lines 30-32).
\end{case}

\begin{claim}\label{message source verifaication 1}
    If in round $r$, agent $i$ receives a message in which $t_B$ is $(j, m)$ and $j\in T_i^m$ or $j=i$, then $t_A$ of the message must already be in $HS_i$ when round $r$. 
\end{claim}

\begin{claim}\label{message source verifaication 2}
    If in round $r$, agent $i$ receives a message in which $t_B$ is $(j, r-1)$ from $k$, then $Type(NS_k^{r-1}[l_{kj}^{r-1}])$ must be $Msg\mbox{-}R$.
\end{claim}

In phase 3, for a link $l_{kp}$, $i$ needs to detect whether there is an inconsistency firstly (line 26). An inconsistency detected in phase 3 may be because

\begin{itemize}
    \item[1.] (message format verification) the format of $NS_j^{r-1}[l_{kp}]$ is incorrect;
    \item[2.] (message source verification) $NS_j^{r-1}[l_{kp}]$ violates Claim \ref{message source verifaication 1} or Claim \ref{message source verifaication 2};
    \item[3.] (random number verification) if the type of $NS_j^{r-1}[l_{kp}]$ is $Msg\mbox{-}R$, the message random number in $NS_j^{r-1}[l_{kp}]$ is different from that in \Call{Random}{}, or if $Msg\mbox{-}X$, the faulty random numbers in \Call{X-Random}{} are different from the random numbers at the corresponding indexes of the sorted set in $NS_j^{r-1}[l_{kp}]$;
    \item[4.] (round number verification) $NS_j^{r-1}[l_{kp}]$ violates one of the claims from Claim \ref{first round number verification} to Claim \ref{last round number verification}.
\end{itemize}

If $i$ detects an inconsistency, then it decides $\perp$ (line 27). If not, $i$ updates the states as previously discussed.

\begin{algorithm*}
    \caption{agent $i$ verifies and updates link state in round $r$}\label{algorithm 2}
    \begin{algorithmic}[1] %每行显示行号
        \Require $r\gets round$, $i\gets id$, $NS^r_i$, $HS_i$, $\{NS^{r-1}\}$, $\{random^r\}$ 
        \Ensure $(NS^r_i, HS_i)$ \textbf{or} $decided$
        \Function {VerifyAndUpdate}{$r$, $i$, $NS^r_i$, $HS_i$, $\{NS^{r-1}\}$, $\{random^r\}$}
            \State $T\gets \Call{IDs}{\{NS^{r-1}\}}$
            \Comment{The function IDs returns ids from NS set}
            \State $S\gets N-T-\{i\}$
                
            \State
            \State Phase 1: update the state of direct links
            \For{$j \in T$}
                \State $NS^r_i[l_{ij}]\gets ((r, i, random_j^r), \phi)$ \Comment{Message source verification is not required if $\phi$}
                \State\Call{AppendHS}{$HS_i, l_{ij}, (r, i, random_j^r)$} \Comment{append the state into HS or decide $\perp$}
            \EndFor
            \For{$j \in S$}
                \If{$\Call{Type}{NS^r_i[l_{ij}]} = Msg\mbox{-}X$}
                    \State $continue$
                \Else
                    \State $NS^r_i[l_{ij}]\gets ((X, r, i, X\mbox{-}random_i^r[l_{ij}]), \phi)$
                    \State\Call{AppendHS}{$HS_i, l_{ij}, (X, r, i, X\mbox{-}random_i^r[l_{ij}])$}
                \EndIf
            \EndFor
                
            \State
            \State Phase 2: verify message chain
            \For{$j \in T$}
                \State \Call{VerifyMsgChain}{$NS_j^{r-1}$} \Comment{Message chain verification or decide $\perp$}
            \EndFor
            \State
            \State Phase 3: verify and update
            \For{$j \in T$}
                \For{$k = 1 \to n-1$}
                    \For{$p = k+1 \to n$}
                        \State $recvState\gets NS_j^{r-1}[l_{kp}].state$
                        \State $localState\gets NS_i^r[l_{kp}].state$
                        \If{an inconsistency is detected}
                            \State$\textbf{Decide}(\perp)$ \Comment{Punishment}
                        \EndIf
                        \If{$\Call{Type}{recvState}=Msg\mbox{-}O$}\Comment{Case 10}
                            \State$Continue$
                        \ElsIf{$\Call{Type}{localState}=Msg\mbox{-}O$}\Comment{Case 11}
                            \State$NS_i^r[l_{kp}]\gets (recvState, (j,r))$
                            \State \Call{AppendHS}{$HS_i, l_{kp}, recvState$}
                        \Else
                            \State$lr\gets localState.round$
                            \State$li\gets localState.id$
                            \State$rr\gets recvState.round$
                            \State$ri\gets recvState.id$
                            \If{($k=i$ \textbf{or} $p=i$)}\Comment{Direct link}
                                \If{$\Call{Type}{localState}=Msg\mbox{-}R$ \textbf{and} $\Call{Type}{recvState}=Msg\mbox{-}R$}\Comment{Case 1}
                                    \State \Call{AppendHS}{$HS_i, l_{kp}, recvState$}
                                \ElsIf{$\Call{Type}{localState}=Msg\mbox{-}R$ \textbf{and} $\Call{Type}{recvState}=Msg\mbox{-}X$}\Comment{Case 2}
                                    \State$\textbf{Decide}(\perp)$
                                \ElsIf{$\Call{Type}{localState}=Msg\mbox{-}X$ \textbf{and} $\Call{Type}{recvState}=Msg\mbox{-}R$}\Comment{Case 3}
                                    \State \Call{AppendHS}{$HS_i, l_{kp}, recvState$}
                                \ElsIf{$\Call{Type}{localState}=Msg\mbox{-}X$ \textbf{and} $\Call{Type}{recvState}=Msg\mbox{-}X$}\Comment{Case 4 and 5}
                                    \If{$li=i$}\Comment{Case 4}
                                        \If{$rr = lr$ \textbf{or} $rr = lr+1$}
                                            \State\Call{AppendHS}{$HS_i, l_{kp}, recvState$}
                                        \ElsIf{$rr = lr - 1$}
                                            \State$NS_i^r[l_{kp}]\gets (recvState,(j,r))$
                                            \State\Call{AppendHS}{$HS_i, l_{kp}, recvState$}
                                        \EndIf
                                    \EndIf
                                \EndIf
\algstore{algorithm2}
    \end{algorithmic}
\end{algorithm*}    
\begin{algorithm*}
    \begin{algorithmic}[1]
    \algrestore{algorithm2}
        %the second page           
                            \Else\Comment{Indirect link}
                                \If{$\Call{Type}{localState}=Msg\mbox{-}R$ \textbf{and} $\Call{Type}{recvState}=Msg\mbox{-}R$}\Comment{Case 6}
                                    \If{$rr \leq lr$}
                                        \State\Call{AppendHS}{$HS_i, l_{kp}, recvState$}
                                    \ElsIf{$rr > lr$}
                                        \State$NS_i^r[l_{kp}]\gets (recvState,(j,r))$
                                        \State\Call{AppendHS}{$HS_i, l_{kp}, recvState$}
                                    \EndIf
                                \ElsIf{$\Call{Type}{localState}=Msg\mbox{-}R$ \textbf{and} $\Call{Type}{recvState}=Msg\mbox{-}X$}\Comment{Case 7}
                                    \State$NS_i^r[l_{kp}]\gets (recvState,(j,r))$
                                    \State\Call{AppendHS}{$HS_i, l_{kp}, recvState$}
                                \ElsIf{$\Call{Type}{localState}=Msg\mbox{-}X$ \textbf{and} $\Call{Type}{recvState}=Msg\mbox{-}R$}\Comment{Case 8}
                                    \State\Call{AppendHS}{$HS_i, l_{kp}, recvState$}
                                \ElsIf{$\Call{Type}{localState}=Msg\mbox{-}X$ \textbf{and} $\Call{Type}{recvState}=Msg\mbox{-}X$}\Comment{Case 9}
                                    \If{$ri \neq li$}
                                        \If{$lr \leq rr$}
                                            \State\Call{AppendHS}{$HS_i, l_{kp}, recvState$}
                                        \Else
                                            \State$NS_i^r[l_{kp}]\gets (recvState,(j,r))$
                                            \State\Call{AppendHS}{$HS_i, l_{kp}, recvState$}
                                        \EndIf
                                    \EndIf
                                \EndIf
                            \EndIf
                        \EndIf
                    \EndFor
                \EndFor                    
            \EndFor                            

            \State \Return{$NS^r_i, HS_i$}
        \EndFunction
    \end{algorithmic}
\end{algorithm*}

\subsection{Proof of the Protocol}
The proof assumes $n>2t+1$. Some variables are defined as follows.

%后续可以写成definition
\noindent $- \  State_i[l_{ij}^r]$ denotes the detection result of agent $i$ on the state of direct link $l_{ij}$ in round $r$. The type of $State_i[l_{ij}^r]$ must be $Msg\mbox{-}R$ or $Msg\mbox{-}X$.

\noindent $- \ \mathcal{C}_i^j(m_1,m_2)$ denotes the agent chain (or we can call it message propagation path) from agent $i$ to $j$. $i$ detects a direct link state in round $m_1$ and sends it to agent $k\neq i,j$ in round $m_1+1$. Then $k$ also sends the state to another agent in round $m_1+2$. And so on, finally, $j$ receives the state in round $m_2$. 

\noindent $- \  Nf^r$ denotes the set of nonfaulty agents in round $r$.

\noindent $- \  F^r$ denotes the set of faulty agents in round $r$.

\noindent $- \  F^{\Delta r}$ denotes the set of faulty agents newly detected in round $r$.

\noindent $- \  x_r$ denotes the number of risk agents in round $r$.

We first prove the upper bound of message passing time and give the round complexity of the algorithm.

\begin{theorem}[\textbf{Message Passing Mechanism}]\label{theorem 1}
    If $i,j \in Nf^{r+t+1}$, all link states in round $r$ can be reached a consensus between $i$ and $j$ at the latest in round $r+t+1$.
\end{theorem}
\begin{proof}
    Consider the state of $link_{kp}$ in round $r$, where $k,p\in N$. Specially, we can consider the messages sent by $k$ and $p$ to be independent of each other and this does not affect the final consensus outcome. For example, $Type(State_k[l_{kp}^r])=Msg\mbox{-}X$ and it is received by all nonfaulty agents in round $m_1$ ($<r+t+1$), and $Type(State_p[l_{kp}^r])=Msg\mbox{-}R$ and it is received by all nonfaulty agents in round $m_2$ ($m_1<m_2<r+t+1$). Even if the detection result of $p$ may be no longer forwarded after round $m_1$, but we still have the correct consensus state in round $r+t+1$ when we consider two detection results independently. We have following cases:
    \begin{itemize}
    \item \textit{Case 1.} $k$ and $p$ are good agents in round $r$. In round $r+1$, $k$ and $p$ send their detection results to all good agents. So if $t=0$, all agents reach a consensus on the state of $l_{kp}$ in round $r+1$. If $t>0$, all nonfaulty agents reach a consensus in round $r+2$. Therefore, all link states of round $r$ among good agents can be reached a consensus in round $r+t+1$.
    \item \textit{Case 2.} $k$ is a risk agent or faulty agent and $p$ is not equal to $k$. Generally, since a receive omission can be converted to a send omission, then each risk agent and faulty agent has at most the following 3 choices when sending messages in each round:
        \begin{itemize}
        \item[1.] It has sending omissions with all other agents. 
        \item[2.] It does not have sending omissions with at least a good agent.
        \item[3.] It has sending omissions with all good agents and no sending omissions with some risk agents or faulty agents.
        \end{itemize}
        Hence, $k$ has three choices in round $r+1$. 
        \begin{itemize}
        \item \textit{Case 2.1.} $k$ chooses 1. Then all agents do not know $State_k[l_{kp}^r]$. All nonfaulty agents agree on the ``unknown-state".
        \item \textit{Case 2.2.} $k$ chooses 2. Then there must be some good agents knowing $State_k[l_{kp}^r]$ in round $r+1$. And all good agents and $k$ know the state in round $r+2$. If $t=1$, $k$ is the only risk agent, then there is a consensus on the state in round $r+2$. But if $t>1$, all nonfaulty agents receive $State_k[l_{kp}^r]$ in round $r+3$ because all good agents must send it to all nonfaulty agents in this round. Thus the lemma holds.
        \item \textit{Case 2.3.} $k$ chooses 3. So no good agents know $State_k[l_{kp}^r]$ in round $r+1$ and $k$ is detected faulty in round $r+1$. Suppose that there is only a risk agent receiving the state. Since agents are independent of each other, it is easy to scale the number of the agents from one to many. We can also divide this case into two cases.
        \begin{itemize}
        \item \textit{Case 2.3.1.} $k$ only sends messages to $p$ in round $r+1$. It means that $Type(State_k[l_{kp}^r])=Msg\mbox{-}R$. If $Type(State_p[l_{kp}^r])=Msg\mbox{-}X$, $p$ does not receive messages from $k$ in round $r+1$. Then the result is the same as that in case 2.1. But if $Type(State_p[l_{kp}^r])=Msg\mbox{-}R$, $k$ has no influence on the final result and the consensus result of $l_{kp}^r$ depends on the choice of $p$. If $p$ also chooses Case 2.3.1, then two states of $l_{kp}^r$ only exist in $k$ and $p$. The final result is also the same as that in case 2.1.
        \item \textit{Case 2.3.2.} $k$ has no sending omissions with some risk agents or faulty agents other than $p$. Then in round $r+2$, the risk agents and faulty agents that have received messages from $k$ also have 3 choices. Take one of the risk agents $l$ as an example. If $l$ chooses 1 or 2 in round $r+2$, the results are the same as those in case 2.1 and case 2.2 where the lemma holds. And if $l$ chooses 3, no good agents know $State_k[l_{kp}^r]$ in round $r+2$. Suppose that when risk and faulty agents choose 3, they must send $State_k[l_{kp}^r]$ to risk agents or faulty agents other than the source agent of the state because if they only send the state back to the source agent, the final results depend only on the source agent, not on themselves. Then until round $r+t$, if from round $r+1$ to round $r+t-1$, all risk agents and faulty agents that have received $State_k[l_{kp}^r]$ choose 3, then the risk (or faulty) agent $z$ in round $r+t$ must be the last risk (or faulty) agent in system. At this time, $z$ has only 2 choices: 1 and 2. And it is easy to get that $State_k[l_{kp}^r]$ must be consensus in round $r+t+1$. But if from round $r+1$ to round $r+t-1$, some risk agents or faulty agents that have received $State_k[l_{kp}^r]$ choose 1 or 2, then the result are the same as those in case 2.1 and case 2.2. 
        \end{itemize}
        \end{itemize}
    \end{itemize}
    In summary, the lemma holds.
\end{proof}

\begin{corollary}\label{corollary 8.1}
    If $i,j \in Nf^{r+x_r+1}$, all link states in round $r$ can be reached a consensus between $i$ and $j$ in round $r+x_r+1$.
\end{corollary}
\begin{proof}
    There are $t-x_r$ faulty agents in round $r$. Since the faulty agents before round $r$ do not send any messages in round $r+2$, it is equivalent to case 2.1 that $k$ sends messages to these faulty agents in round $r+1$. Then the total number of risk and faulty agents in case 2.3 can be reduced to $x_r$. Therefore, in case 2.3.2, if keeping choosing 3, there are no risk or faulty agents anymore up to round $r+x_r$ and reach a consensus on all link states of round $r$ between $i$ and $j$ in round $r+x_r+1$.
\end{proof}

\begin{lemma}[\textbf{Round Complexity}]\label{lemma 10}
    The link states $HS$ of the second clean round and all previous rounds can be reached a consensus among all nonfaulty agents at the latest in round $t+3$.
\end{lemma}
\begin{proof}
    By theorem \ref{theorem 1}, the smaller of the round $r$, the smaller the supremum of the round in which the link states in round $r$ can be reached a consensus. Hence, we directly consider the second clean round. Suppose the second clean round is $y$. Then there are already at least $y-2$ faulty agents in round $y$. That is $x_y\leq t-y+2$. By corollary \ref{corollary 8.1}, the link states in round $y$ can be consensus in round $y+x_y+1$. Since $y+x_y+1\leq t+3$, the lemma holds.
\end{proof}

Then it is proved that the algorithm satisfies all the properties of uniform consensus with general omission failures.

\begin{lemma}\label{lemma 1.1}
    If $i$ is a nonfaulty agent, then $Type(HS_i^{r\sim t+3}[l_{kp}])=Msg\mbox{-}X$ when $Type(HS_i^r[l_{kp}])=Msg\mbox{-}X$ ($k,p\in N$).
\end{lemma}
\begin{proof}
    Link $l_{kp}$ can not recover after a fault occurs. So if $l_{kp}$ is a faulty link in round $r$, then its state must also be $Msg\mbox{-}X$ in subsequent rounds. Moreover, $HS$ also expands all $Msg\mbox{-}X$ states backwards in \Call{LastUpdate}{}.
\end{proof}

\begin{lemma}\label{lemma 1.2}
    If $i$ is a nonfaulty agent, then $Type(HS_i^{1\sim r-1}[l_{kp}])=Msg\mbox{-}R$ when $Type(HS_i^r[l_{kp}])=Msg\mbox{-}R$ ($k,p\in N$).
\end{lemma}
\begin{proof}
    Since $Type(HS_i^r[l_{kp}])\neq Msg\mbox{-}O$, there must be an agent in $k$ or $p$ (supposing $p$) that has reported $State_p[l_{kp}^r]$ in round $r+1$, and finally the state has been transmitted to $i$. We suppose that $i$ receives the state in round $r\overset{'}{}$. Then we have $\mathcal{C}_p^i(r,r\overset{'}{})$. Since link omission is irreversible, $\mathcal{C}_p^i(r,r\overset{'}{})$ must be nonfaulty from round 1 to round $r-1$. Hence, we have $State_p[l_{kp}^{1\sim r-1}]$ must eventually be received by $i$. That means $Type(HS_i^{1\sim r-1}[l_{kp}]) \neq Msg\mbox{-}O$. Combining lemma \ref{lemma 1.1}, it is easy to get $Type(HS_i^{1\sim r-1}[l_{kp}]) \neq Msg\mbox{-}X$. Thus the lemma holds.
\end{proof}

\begin{lemma}\label{lemma 1.3}
    If $i$ is a nonfaulty agent, then $Type(HS_i^{r\sim t+3}[l_{kp}])=Msg\mbox{-}O$ when $Type(HS_i^r[l_{kp}])=Msg\mbox{-}O$ ($k,p\in N$).
\end{lemma}
\begin{proof}
    For a contradiction, let $Type(HS_i^{m_1}[l_{kp}])\neq Msg\mbox{-}O$ ($m_1 > r$) when $Type(HS_i^r[l_{kp}])=Msg\mbox{-}O$. Suppose that $i$ receives the state of $link_{kp}^{m_1}$ in round $m_2$ ($m_2 > m_1$). Let's suppose $i$ receives it from $k$. Then we must have $\mathcal{C}_k^i(m_1,m_2)$. Since link omission is irreversible, the message propagation path is also correct for round $r$, so that $i$ must receive $State_k[l_{kp}^r]$ and then $Type(HS_i^r[l_{kp}])\neq Msg\mbox{-}O$. Therefore, we have a contradiction here and the lemma holds.
\end{proof}

\begin{lemma}\label{lemma 2}
    A nonfaulty agent must have correct links with at least $n-t-1$ agents other than itself in a round.
\end{lemma}
\begin{proof}
    We can know that, for a nonfaulty agent $i$, $|lost_i|$ must be less than or equal to $t$. So it is easy to get that $i$ have correct links with at least $n-t-1$ agents.
\end{proof}

\begin{lemma}\label{lemma 3}
    A nonfaulty agent must have correct links with at least 2 good agents other than itself in a round.
\end{lemma}
\begin{proof}
    We analyze the nonfaulty agent $i$ from two aspects of good agent and risk agent.
    \begin{itemize}
    \item \textit{Case 1.} $i$ is a good agent. Then $i$ must have correct links with all other $n-t-1$ good agents. Since $n>2t+1$ and $n\geq 3$, there must be $n-t-1\geq 2$. 
    \item \textit{Case 2.} $i$ is a risk agent. Suppose that $i$ is nonfaulty in round $r$ and there are $f$ faulty agents in this round. Because it must remove faulty agents when computing the state of $i$, combining lemma \ref{lemma 2}, the risk agent $i$ needs to have correct links with at least $(n-t-1)-(t-f-1)=n-2t+f$ good agents in round $r$. Since $n-2t>1$, $n-2t+f>f+1$. Then $n-2t+f\geq 2$ always holds.
    \end{itemize}
    Therefore, the lemma holds.
\end{proof}

\begin{remark}
    For a faulty agent $f$ in round r, since it has faulty links with more than $t$ agents in round $r$, then it does not send messages to any agents after at most 2 rounds. Hence, we claim that in round $r$ and later, the faulty agent $f$ needs to be removed when computing the number of connections of other agents.
\end{remark}

\begin{lemma}\label{lemma 5}
    Supposing that the direct link state information of $j$ in round $r$ can be agreed by all nonfaulty agents in round $m$. If $i$ is a nonfaulty agent in round $m$ and agent $j$ is considered to be a uncertain agent in $HS_i^r$, then $j$ must be a faulty agent in $HS_i^{r+1}$.
\end{lemma}
\begin{proof}
    The proof argument is by contradiction. Assume that $j$ is considered to be a nonfaulty agent or a uncertain agent in $HS_i^{r+1}$.
    \begin{itemize}
    \item \textit{Case 1.} $j$ is a nonfaulty agent in $HS_i^{r+1}$ when it is considered to be a uncertain agent in $HS_i^r$. $j$ must send $NS_j^r$ to at least 2 good agents in round $r+1$ (lemma \ref{lemma 3}). Then these good agents send the direct link states of $j$ to all nonfaulty agents. Hence, $j$ must be a certain agent in $HS_i^r$. A contradiction.
    \item \textit{Case 2.} $j$ is a uncertain agent in $HS_i^{r+1}$ when it is considered to be a uncertain agent in $HS_i^r$. It is easy to see that the link states between $j$ and good agents can not be unknown-state in round $r$ and $r+1$. Since the number of good agents $n-t$ must be greater than $t$, $j$ can not have faulty links with all good agents. Then it must send $NS_j^r$ to some good agents in round $r+1$. Equally, $j$ must be a certain agent in $HS_i^r$. A contradiction. 
    \end{itemize}
    Thus we reach contradictions in all cases, which proves the lemma.
\end{proof}

\begin{lemma}\label{lemma 6}
    If round $r$ is a clean round, the state of $l_{ij}^{r-1}$ can be reached a consensus by all nonfaulty agents in round $r+2$, where $i,j \in Nf^{r-1}$ and $j\neq i$.
\end{lemma}
\begin{proof}
    We can pay attention to the state of $l_{ij}^{r-1}$. Consider the following cases:
    \begin{itemize}
    \item \textit{Case 1.} $i$ and $j$ are good agents. By lemma \ref{lemma 3}, a risk agent must have correct links with some good agents. Hence, $i$ sends $State_i[l_{ij}^{r-1}]$ to all good agents and risk agents having correct links with $i$ in round $r$. And $j$ also does this. In round $r$ all good agents have two detection results of $l_{ij}^{r-1}$. Then after updating, they send the uniform state to all risk agents that have faulty links with $i$ and $j$ in round $r+1$.
    \item \textit{Case 2.} $i$ and $j$ are risk agents. $i$ and $j$ send their detection results of $l_{ij}^{r-1}$ to some good agents (denoted by $U$) and risk agents in round $r$. Then two results are sent to all good agents by the agents in $U$ in round $r+1$. So every good agent knows the uniform state of $l_{ij}^{r-1}$ in round $r+1$. Therefore, all nonfaulty agents reach a consensus on the state in round $r+2$.
    \item \textit{Case 3.} $i$ is a good agent and $j$ is a risk agent. Similarly, it is easy to get that all good agents have the uniform state of $l_{ij}^{r-1}$ in round $r+1$ by case 1 and case 2. So this is what we want.
    Thus the lemma holds.
    \end{itemize}
\end{proof}

\begin{lemma}\label{lemma 7}
    If round $r$ is a clean round, then in the $HS_i^{r-1}$ ($i\in Nf^{r+2}$), the state of link $l_{kp}$ ($k\in Nf^{r-1}$ and $p\in N$) can not be $Msg\mbox{-}O$.
\end{lemma}
\begin{proof}
    Assuming, without influence, that the messages of $p$ have no effect on the state of $l_{kp}$. Since $k$ is a nonfaulty agent, by lemma \ref{lemma 6}, $State_k[l_{kp}^{r-1}]$ is received by all nonfaulty agents in round $r+2$. Hence $i$ must know the state of $l_{kp}^{r-1}$. The lemma holds.
\end{proof}

\begin{corollary}
    If round $r$ is a reliable round and the total number of rounds is greater than $r+3$, there are not uncertain agents in round $r$.
\end{corollary}
\begin{proof}
    For a contradiction, let $i$ be a uncertain agent in round $r$, by lemma \ref{lemma 5}, we have two cases. Both case 1 and case 2, by lemma \ref{lemma 6}, there are contradictions to the assumption. Hence, $i$ must be a faulty agent in $HS^{r+1}$. The unknown-state link is regarded as correct link so that $i$ is regarded as a nonfaulty agent in $HS^r$. Then it is a contradiction to the assumption that $r$ is a reliable round. Thus the lemma holds.
\end{proof}

\begin{lemma}\label{lemma 9}
    There is at least one clean round in $t+1$ rounds.
\end{lemma}
\begin{proof}
    Suppose, for a contradiction, that there is no clean round in $t+1$ rounds. Then there must be new faulty agents added in each round. So there are at least $t+1$ faulty agents in $t+1$ rounds. This contradicts to the assumption that there are at most $t$ faulty agents.
\end{proof}

\begin{corollary} \label{corollary 9.1}
    There are at least two clean rounds in $t+2$ rounds.
\end{corollary}

\begin{corollary}\label{corollary 9.2}
    In $t+2$ rounds, there must be one reliable round $r$ in which at most one new faulty agent is detected.     
\end{corollary}
\begin{proof}
    Suppose that there are $a$ clean rounds in $t+2$ rounds. We prove the lemma from two cases:
    \begin{itemize}
    \item \textit{Case 1.} All clean rounds are greater than round 1. And for a contradiction, two new faulty agents are detected in each reliable round. Then there are $2a$ faulty agents and there are still $t-2a$ faulty agents remaining in $t+2-2a$ rounds. It is easy to see that $t+2-2a>t-2a$. Therefore, there must be clean rounds in the remaining $t+2-2a$ rounds. This contradicts to the assumption that there are $a$ clean rounds in $t+2$ rounds.
    \item \textit{Case 2.} Round 1 is a clean round. And for a contradiction, two new faulty agents are detected in each reliable round. Then there are $2(a-1)$ faulty agents and there are still $t-2a+2$ faulty agents remaining in $t+2-2a+1$ rounds. Since $t+2-2a+1>t-2a+2$, there must be clean rounds in the remaining $t+2-2a+1$ rounds. A contradiction.
    \end{itemize}
    Thus we reach a contradiction in every case, which proves the lemma.
\end{proof}

\begin{lemma}\label{lemma 11}
    In round $t+3$, if a faulty agent $i\in F^{\Delta t+3}$ can receive messages from at least one good agent, the link states of the second clean round and all previous rounds can also be reached a consensus among $i$ and all nonfaulty agents.
\end{lemma}
\begin{proof}
    By corollary \ref{corollary 9.2}, from the second clean round $y$ to round $t+3$, there must be a reliable round (suppose the first is $r$) in which at most one new faulty agent is detected, because the total number of risk and faulty agents is $t-y+2$ and the total number of rounds is $t-y+4$. Suppose $r+1 = y\overset{'}{}$ which is a clean round. Since $F^{\Delta t+3}\neq \emptyset$, so $y<y\overset{'}{}<t+3$. Since no new faulty agents are detected in round $y\overset{'}{}$, risk agents can only choose 2 in round $y\overset{'}{}$ (See details in theorem \ref{theorem 1}). Hence, in round $y\overset{'}{}+1$, all good agents reach a consensus on the link states $HS$ of round $y$ and before rounds. We divide $y\overset{'}{}$ into two cases to prove as follows:
    \begin{itemize}
    \item \textit{Case 1.} $y<y\overset{'}{}<t+2$. Then we have $y\overset{'}{}+2\leq t+3$. In round $y\overset{'}{}+2$, all good agents send the latest and uniform link states of round $y$ and before rounds to all agents. Thus $i$ must reach a consensus.
    \item \textit{Case 2.} $y\overset{'}{}=t+2$. We assume that for the reliable rounds in which two or more faulty agents are detected, the faulty agents can be averaged to the next round and then the clean round can also be regarded as a normal round. Then it can be seen that the number of faulty agents keeps increasing in each round from round $y+1$ to round $t+1$. Thus at least $t+1-y$ faulty agents have been added until round $y\overset{'}{}$. Since there are $x_y$ ($\leq t-y+2$) risk agents in round $y$, then at most one risk agent remains in round $y\overset{'}{}$ and it must be $i$. Then it is easy to get see $i$ must reach a consensus in round $t+3$.
    \end{itemize}
    Thus, the lemma holds.
\end{proof}

\begin{theorem}\label{theorem 2}
    \Call{Consensus}{} solves uniform consensus if at most $t$ agents omit to send or receive messages, $n>2t+1$, suppose that all agents are honest.
\end{theorem}
\begin{proof}
    Since $n>2t+1$, it is easy to see that no inconsistency is detected.
    
    \textbf{Termination:}  From Algorithm \ref{algorithm 1}, nonfaulty agents must decide in round $t+4$ and faulty agents decide before round $t+4$.
    
    \textbf{Validity:} Since no inconsistency is detected, all agents make decisions different from $\perp$. For agent $i$, if $i$ decides a value $decision_i$, $decision_i$ must be the initial preference of an agent in $decision \ set \ D$. Since $D$ depends on $HS_i^{m\overset{*}{}}$, it must have $D\subseteq N$. Therefore, $decision_i$ satisfies the validity property. If $i$ decides $\parallel$, which means $i$ has no decision and $\parallel$ does not affect the final consensus outcome. Thus it also conforms to the validity.
    
    \textbf{Uniform Agreement:} We prove this from the following cases:
    \begin{itemize}
    \item \textit{Case 1.} Agent $i$ and $j\in Nf^{t+4}$. By corollary \ref{corollary 9.1}, there must be a $decision \ round \ m^*$ in $t+3$ rounds. And by lemma \ref{lemma 10}, we have $HS_i^{1\sim m^*+1}=HS_j^{1\sim m^*+1}$. Then $D_i=D_j=D^*$. Since the pieces of preferences and proposals of all the agents in $D^*$ must be saved by at least 2 good agents (lemma \ref{lemma 3}), all good agents and some risk agents can restore all initial values and proposals of the agents in $D^*$ in round $t+3$. We denote these agents by $Nf^{t+4}_1$ and $Nf^{t+4}_2 = Nf^{t+4}\setminus Nf^{t+4}_1$. Thus all agents in $Nf^{t+4}_1$ have the same set $C$ and give a unified $consensus$ set containing one value. That is if agent $u$ and $v \in Nf^{t+4}_1$, then there must be $consensus_u=consensus_v=\{cons\}$ and $|consensus_u|=|consensus_v|=1$ in round $t+3$. And if agent $w\in Nf^{t+4}_2$, $consensus_w=\emptyset$ in round $t+3$. 
    \item \textit{Case 2.} We analyze the agents in $F^{t+4}$. It is easy to see that $F^{t+4}=F^{t+2}\cup F^{\Delta t+3} \cup F^{\Delta t+4}$.
    \begin{itemize}
    \item \textit{Case 2.1.} $i\in F^{t+2}$. $i$ must decide $\parallel$ at the latest in the receive phase of round $t+3$.
    \item \textit{Case 2.2.} $i\in F^{\Delta t+3}$. Suppose that $F^{\Delta t+3}_1$ denotes the set of faulty agents that can receive the messages from some good agents in round $t+3$ and send messages in round $t+4$. Then $F^{\Delta t+3}_2 = F^{\Delta t+3}\setminus F^{\Delta t+3}_1$. It is easy to see that the agents in $F^{\Delta t+3}_2$ definitely do not send messages in round $t+4$ and they must decide $\parallel$ in round $t+3$. If $i\in F^{\Delta t+3}_1$, by lemma \ref{lemma 11}, $D_i$ must be same as $D^*$ of good agents in case 1. Thus if $i$ can restore all initial preferences and proposals in $D_i$, it must have $consensus_i=\{cons\}$ and $|consensus_i|=1$ in round $t+3$. Otherwise, $consensus_i=\emptyset$ in round $t+3$.
    \item \textit{Case 2.3.} $i\in F^{\Delta t+4}$. Since $i$ is a nonfaulty agent in round $t+3$, $consensus_i$ has the same two possible states in round $t+3$ as in case 1. Suppose $F^{\Delta t+4} = F^{\Delta t+4}_1 \cup F^{\Delta t+4}_2$. $F^{\Delta t+4}_1$ denotes the set of faulty agents that can not detect faulty by itself in the receive phase of round $t+4$. $F^{\Delta t+4}_2$ denotes the set of faulty agents that can detect that they become faulty agents in the receive phase of round $t+4$. Therefore, the agents in $F^{\Delta t+4}_1$ decide $cons$ by $consensus_i$ and the agents in $F^{\Delta t+4}_2$ decide $\parallel$ in round $t+4$.
    \end{itemize}
    \end{itemize}
    In summary, $consensus$ set only has two types in round $t+3$ and round $t+4$: $\{cons\}$ and $\emptyset$. The agents in $Nf^{t+4}_1 \cup Nf^{t+4}_2 \cup F^{\Delta t+4}_1$ decide $cons$ in round $t+4$, the agents in $F^{t+2} \cup F^{\Delta t+3}$ decide $\parallel$ before round $t+4$, and the agents in $F^{\Delta t+4}_2$ decide $\parallel$ in round $t+4$. Thus Uniform Agreement holds.
\end{proof}

To achieve the Nash equilibrium, we make some appropriate assumptions about initial preferences and failure patterns. Failure pattern represents a set of failures that occur during an execution of the consensus protocol \cite{2016PODC}. Specifically, we assume that initial preferences and failure patterns are $blind$.

\begin{definition}\label{definition 1}
     The blind initial values mean that each agent cannot guess the preferences of other agents and the probability of its own preference becoming the consensus cannot be improved by trusting others. 
\end{definition}

By definition \ref{definition 1}, we can get that if an agent wants to improve its own utility, it can only rely on itself, for example, increasing the probability of entering the $decision \ set$ and reducing the number of agents in the $decision \ set$ and so on.

\begin{definition}\label{definition 2}
    The blind failure patterns mean that before $t$ faulty agents appear, an agent cannot guess the link states in the following rounds. Then we have that
    \begin{itemize}
    \item[1.] if agent $i$ does not know the link states of round $m$ in round $r$ and $j$ is a nonfaulty agent in round $r$, then

    $P(i \in F^{\Delta m} \ | \ link_{ij} \ is \ faulty) = P(j \in F^{\Delta m} \ | \ link_{ij} \ is \ faulty) \leq \alpha$;
    \item[2.] for round $m_1$ and $m_2$, if $m_1<m_2$ and $i$ does not know the link states of round $m_2$, then
    
    $P(v_i \ becomes \ consensus \ | \ m_1 \ is \ the \ decision \ round) = P(v_i \ becomes \ consensus \ | \ m_2 \ is \ the \ decision \ round)$;
    \item[3.] for $t+1$ rounds, if the link states of each round in $t+1$ rounds are unknown to agent $i$, then for a round $r$ in $t+1$ rounds,
    $P(r \ is \ a \ clean \ round) \geq \frac{1}{t+1}$.
    \end{itemize}
\end{definition}

\begin{theorem}
    If $n>2t+1$, and at most $t$ agents have omission failures at the same time, and agents prefer consensus, and failure patterns and initial preferences are blind, then $\Vec{\sigma}^{\Call{Consensus}{}}$ is a $Nash \ equilibrium$.
\end{theorem}
\begin{proof}
    To prove $Nash \ equilibrium$, we need to show that it is impossible for each agent $i\in N$ to increase its utility $u_i$ with all possible deviations $\sigma_i$. That is proving that for each agent $i$, there must be
    \begin{equation}\label{equ 1}
        u_i(\sigma_i^{\Call{Consensus}{}},\sigma_{-i}^{\Call{Consensus}{}}) \geq u_i(\sigma_i, \sigma_{-i}^{\Call{Consensus}{}}).
    \end{equation}
    We use the same deduction method as in \cite{2016PODC}. Consider all the ways that $i$ can deviate from the protocol to affect the outcome as follows:
    \begin{itemize}
    \item[1.] $i$ generates a different value $v_i\overset{'}{} \neq v_i$ (or $proposal_i^{'} \neq proposal_i$) and sends $q_i\overset{'}{}(j)$ (or $b_i^{'}(j)$) to some agents $j \neq i$.
    \item[2.] $q_i(j)$ (or $b_i(j)$) sent by $i$ can not restore $v_i$ (or $proposal_i$).
    \item[3.] $i$ does not choose $random_i$ or $X\mbox{-}random_i$ or $proposal_i$ appropriately, such as not randomly.
    \item[4.] $i$ sends an incorrectly formatted message to $j \neq i$ in round $m$.
    \item[5.] if $|lost_i|>t$ in round $m$, $i$ does not decide $\parallel$, but continues to execute the following protocol.
    \item[6.] $i$ lies about the state of $l_{kp}$ in round $m$; that is, in round $m$, $i$ sends a state of $l_{kp}$ which is different from $NS_i^{m-1}[l_{kp}]$.
    \item[7.] $i$ sends an incorrect $random$ or $X\mbox{-}random$ of $l_{kp}$ to $j$ in round $m$.
    \item[8.] $i$ sends an incorrect $q_l(i)$ (or $b_l(i)$) to $j \neq l$ different from the $q_l(i)$ (or $b_l(i)$) that $i$ received from $l$ in round 1.
    \item[9.] $i$ sends an incorrect $consensus_i$ to $j \neq i$ in round $t+4$.
    \item[10.] $i$ pretends to crash in round $m$.
    \end{itemize}
    We consider these deviations one by one and prove that $i$ does not gain by any of deviations. That is, (\ref{equ 1}) holds if $i$ deviates from the protocol by these deviations on the list above.
    \begin{itemize}
    \item \textit{Type 1.} (\romannumeral1) If $i$ sends $q_i^{'}(j)$ to some agents, then either an inconsistency is detected because of secret restoring error, or $i$ does not gain. Specifically, if $i$ is the agent whose value is chosen, then $i$ is worse off if it lies than it doesn't, since some agents can not restore $v_i$, but they can if following the protocol; Then if $i$ is not the agent whose preference is chosen, then it does not affect the outcome. (\romannumeral2) $i$ sends $b_i^{'}(j)$. Then an inconsistency is detected if restoring polynomial error or generating different consensus values in the system. And if no inconsistency is detected, then either all agents that receive $b_i$ or $b_i\overset{'}{}$ are faulty or both $b_i$ and $b_i\overset{'}{}$ do not affect the final outcome. Since changing the proposal can not increase $i$'s utility, thus $i$ does not gain. Therefore, both (\romannumeral1) and (\romannumeral2), $i$ does at least as well if $i$ uses the strategy $\sigma_i^{\Call{Consensus}{}}$ as it deviates from the protocol according to the type 1. So (\ref{equ 1}) holds.
    
    \item \textit{Type 2.} It is easy to see that either an inconsistency is detected or no benefit because there is no increase in the probability that the value of $i$ becomes the consensus. Thus (\ref{equ 1}) holds.
    \item \textit{Type 3.} (\romannumeral1) Since other agents follow the protocol, it does not affect the final outcome because the two kinds of random numbers are only used for verification. (\romannumeral2) Since $i$ does not know the proposals of other agents in round 1, using different proposals can not improve the probability that the preference of $i$ becomes the consensus. Thus (\ref{equ 1}) holds.
    \item \textit{Type 4.} If $i$ sends an incorrectly formatted message to $j$, then either an inconsistency is detected by $j$ or it does not affect the outcome since $j$ omits to receive messages from $i$ in round $m$. Thus, $i$ does not gain, so (\ref{equ 1}) holds.
    
    \item \textit{Type 5.} Since $|lost_i|>t$, $i$ does not receive messages from at least $t+1$ agents in round $m$, that is $i$ has receiving omission failures with at least two good agents. 
    \begin{itemize}
    \item \textit{Case 1.} $i$ does not guess message random numbers in round $m+1$. Then by claim \ref{claim 1}, an inconsistency is detected.
    \item \textit{Case 2.} $i$ guesses the message random number in round $m+1$ and has correct links with the remaining $n-t-2$ agents, and these agents are all nonfaulty agents. Then by the claim \ref{claim 1}, $i$ can successfully send messages in round $m+1$ iff $i$ can guess a message random number $random_j^m$ from a nonfaulty agent $j$ and $i$ has no sending omission failures with $j$. That is the random guessing does not change the detection result of the state of $i$ by other agents in round $m$. Clearly the probability that $i$ can guess a random number is $\frac{1}{n}$.
    \begin{itemize}
    \item \textit{Case 2.1}
    If $i$ guesses some random numbers from agents $j$ and has sending omission failures with each agent $j$. Then it does not affect the outcome even if the random numbers are correctly guesses.
    \item \textit{Case 2.2}
    Suppose that $i$ guesses only one random number from the good agent $j$ in a round and $i$ does not have sending omission failures with $j$. If $i$ only guesses the message random number in round $m+1$, then we have that 
    $u_i(\sigma_i, \sigma_{-i}^{\Call{Consensus}{}}) \leq \frac{1}{n}P(the \ decision \ round \ is \ in \ m+1 \ rounds)\frac{1}{|D|}\beta_0 + \gamma_1\beta_1 + \gamma_2\beta_2$. Thus 
    \begin{equation}\label{equ 5.1}
        u_i(\sigma_i, \sigma_{-i}^{\Call{Consensus}{}}) \leq \frac{1}{n}\times\frac{1}{n-t+1}\beta_0 + \gamma_1\beta_1 + \gamma_2\beta_2,
    \end{equation}
    which means that $i$ only guess one random number in round $m+1$ and then $i$ must be in $decision \ set \ D$ ($|D|\geq |G|=n-t$). Since $i\in D$, there are at least $n-t+1$ agents in $D$. It is easy to see that if $i$ guesses random numbers in multiple rounds, the utility must be less than (\ref{equ 5.1}). If $i$ following the protocol, then $i$ must decide $\parallel$ in round $m$. Since $i$ has receiving omission failures in round $m$, the link states after round $m-1$ must be unknown to $i$. Thus by the definition \ref{definition 2} and the assumptions about omission failures, we can get that $P(round \ m-1 \ or \ m \ is \ a \ clean \ round)\geq \frac{1}{t+1}$. Then
    \begin{equation}\label{equ 5.2}
        u_i(\sigma^{\Call{Consensus}{}}) \geq \frac{1}{t+1}\times \frac{1}{|D|_{m-2}}\beta_0 + \gamma^c\beta_1.
    \end{equation}
    Since $n>2t+1$, the (\ref{equ 1}) must hold. 
    \item \textit{Case 2.3}
    If $i$ guesses more than one random numbers in round $m+1$, then the utility of $i$ must be less than (\ref{equ 5.1}). And $i$ does at least as well by following the protocol as the deviation because the guessing work does not affect the state of $i$ in round $m$. Specifically, either if $i$ is a nonfaulty agent detected by other agents, then (\ref{equ 5.2}) holds, or if $i$ is a faulty agent, then it does not affect the outcome even if $i$ guesses the random correctly. Thus (\ref{equ 1}) holds.
    \end{itemize}
    \item \textit{Case 3.} If $i$ has sending omission failures with the remaining $n-t-2$ agents, then either no benefit since $i$ is faulty in round $m$ detected by other agents, or $i$ does not gain if $i$ is nonfaulty because the utility of deviating from the protocol must be less than (\ref{equ 5.1}). 
    \end{itemize}
    In summary, either an inconsistency is detected by claim \ref{claim 1} if $i$ does not guess the message random numbers, or no benefit from guessing the message random numbers. Thus, yet again, (\ref{equ 1}) holds.
    
    \item \textit{Type 6.} By the proof of Type 5, it is easy to see that $i$ must be a nonfaulty agent detected by $i$ itself in round $m-1$. Since there is more than one state of a link, we partition this deviation into eight cases, and show that $i$ does at least as well by using $\sigma_i^{\Call{Consensus}{}}$ as it deviates from the protocol by these eight deviations.
    \begin{itemize}
    \item \textit{Case 1.} $k$ or $p=i$, such as $k=i$, and $Type(NS_i^{m-1}[l_{kp}^r])=Msg\mbox{-}R$ where $r\leq m-1$, and $i$ pretends $Type(State_i[l_{kp}^r])=Msg\mbox{-}X$ in round $m$. 
    \begin{itemize}
    \item \textit{Case 1.1.} $r<m-1$. Then $State_i[l_{kp}^r]$ must be sent in round $r+1$ in order to enable message chain mechanism to succeed. Since $r+1<m$, an inconsistency must be detected in round $m$.
    \item \textit{Case 1.2.} $r=m-1$ and $r$ is the $decision \ round \ m^*$. Since $i$ does not know the link states of round $m-1$ in the sending phase of round $m$, by definition \ref{definition 2}, if $p$ is a nonfaulty agent in round $m$, then $P(i \in F^{\Delta m-1} \ | \ i \ pretends \ link_{ip} \ is \ faulty) = P(p \in F^{\Delta m-1} \ | \ i \ pretends \ link_{ip} \ is \ faulty) \leq \alpha$. Suppose that the $decision \ set$ in round $r$ is $D$ when following the protocol. We can see that $|D|\geq 3$. (\romannumeral1) If all agents become faulty due to the deviation of $i$, then there is no consensus in the system and $i$ does not gain. (\romannumeral2) If the deviation does not cause no solution and $p$ is a nonfaulty agent in round $m-1$, then we have that
    $u_i(\sigma_i, \sigma_{-i}^{\Call{Consensus}{}}) \leq \alpha (1-\alpha)\frac{1}{|D|-1}\beta_0 + \alpha (1-\alpha)\frac{|D|-2}{|D|-1}\beta_1 + (1-\alpha)^2\frac{1}{|D|}\beta_0 + (1-\alpha)^2\frac{|D|-1}{|D|}\beta_1 + \alpha\beta_1$. That is
    \begin{equation}\begin{aligned}\label{equ 2}
    &u_i(\sigma_i, \sigma_{-i}^{\Call{Consensus}{}}) \leq \frac{(1-\alpha)(|D|-1+\alpha)}{(|D|-1)|D|}\beta_0 + \\ &\left[ \alpha(1-\alpha)\frac{|D|-2}{|D|-1} + (1-\alpha)^2\frac{|D|-1}{|D|} + \alpha \right] \beta_1.
    \end{aligned}
    \end{equation}
    It is easy to get that
    \begin{equation}\label{equ 3}
    u_i(\sigma^{\Call{Consensus}{}})= \frac{1}{|D|}\beta_0+\frac{|D|-1}{|D|}\beta_1.
    \end{equation}
    If $\alpha=0$, then $(1-\alpha)(|D|-1+\alpha)$ in (\ref{equ 2}) takes its supremum $|D|-1$. Hence, there must be (\ref{equ 1}). (\romannumeral3) If $p$ is a faulty agent in round $m-1$, then the deviation does not affect the outcome. Thus the case (\romannumeral1),(\romannumeral2) and (\romannumeral3) hold (\ref{equ 1}).
    \item \textit{Case 1.3.} $r=m-1$ and $r=m^*+1$. Since the link states of next $reliable \ round$ are unknown to $i$, either there is no solution since all agents become faulty; the deviation of $i$ does not affect the final outcome; or by definition \ref{definition 2}, $i$ does at least as well by using $\sigma_i^{\Call{Consensus}{}}$ as it deviates from the protocol because the probability that $v_i$ becomes consensus does not increase. Thus (\ref{equ 1}) holds.
    \item \textit{Case 1.4.} $r=m-1$ and $r\neq m^*$ and $r \neq m^*+1$. Either there is no solution since all agents become faulty or there is no benefit because it does not affect the $decision \ round$. 
    \end{itemize}
    In summary, all cases in case 1 cannot make $i$ gain.
    
    \item \textit{Case 2.} $k$ or $p=i$, such as $k=i$, and $Type(NS_i^{m-1}[l_{kp}^r])=Msg\mbox{-}X$ where $r\leq m-1$, and $i$ pretends $Type(State_i[l_{kp}^r])=Msg\mbox{-}R$ in round $m$. If $r < m-1$, $i$ does not gain, which is the same as that in case 1.1. If $r=m-1$, since $i$ is nonfaulty in round $m-1$ and $i$ does not know the link states of round $m-1$, then by definition \ref{definition 2}, there is no benefit in guessing the message random number with the probability $\frac{1}{n}$. So (\ref{equ 1}) holds.
    
    \item \textit{Case 3.} $k$ and $p \neq i$, and $Type(NS_i^{m-1}[l_{kp}^r])=Msg\mbox{-}R$ or $Msg\mbox{-}O$, and $i$ pretends $Type(State[l_{kp}^r])=Msg\mbox{-}X$ in round $m$. Suppose $i$ lies about the detection result $State_k[l_{kp}^r]$ of $k$. (\romannumeral1) If $k$ is faulty in round $r$, it does not affect the final outcome even if no inconsistency is detected. (\romannumeral2) If $k$ is nonfaulty in round $r$, then $k$ must send the faulty random numbers $X\mbox{-}random_k^r$ to at least two good agents (lemma \ref{lemma 3}). And $i$ guesses a faulty random number with the probability $\frac{1}{2}$. If it does not affect the states of $k$ and $p$, then $i$ does not gain even if guessing the random correctly. Otherwise, if the $decision \ round$ is changed, then $u_i(\sigma_i, \sigma_{-i}^{\Call{Consensus}{}})$ must be less than $\frac{1}{2}\times\frac{1}{n-t}\beta_0+\frac{1}{2}\times\frac{n-t-1}{n-t}\beta_1+\frac{1}{2}\beta_2$. And since $i$ is nonfauty in round $m-1$, then $u_i(\sigma^{\Call{Consensus}{}})$ must be greater than $\frac{1}{n}\beta_0+\frac{n-1}{n}\beta_1$. By the definition of $n>2t+1$, (\ref{equ 1}) holds. So (\ref{equ 1}) holds both in (\romannumeral1) and (\romannumeral2).
    
    \item \textit{Case 4.} $k$ and $p \neq i$, and $Type(NS_i^{m-1}[l_{kp}^r])=Msg\mbox{-}X$ or $Msg\mbox{-}O$, and $i$ pretends $Type(State[l_{kp}^r])=Msg\mbox{-}R$ in round $m$. We also suppose that $i$ lies about the detection result $State_k[l_{kp}^r]$ of $k$. (\romannumeral1) If $Type(NS_i^{m-1}[l_{kp}^r])=Msg\mbox{-}O$, then it does not affect the final outcome even if $i$ guesses the random correctly, since $Msg\mbox{-}R$ and $Msg\mbox{-}O$ have the same meaning when computing the state of a agent. (\romannumeral2) If $Type(NS_i^{m-1}[l_{kp}^r])=Msg\mbox{-}X$, then either an inconsistency is detected by message random number verification or link state conflict; it does not affect the outcome if the states of $k$ and $p$ are unchanged after deviating; or it makes round $r$ a $clean \ round$. Then if there is already a $decision \ round \ r^*$ and $r^* \leq r-1$, it does not affect the outcome because we need the first reliable round finally. And if $r^* > r-1$, then the utility of $i$ decreases because $decision \ round$ is advanced compared to following the protocol. If there is no $decision \ round$, by definition \ref{definition 2}, $i$ does at least as well by using $\sigma_i^{\Call{Consensus}{}}$ as it deviates from the protocol. Hence, (\ref{equ 1}) holds again.
    
    \item \textit{Case 5.} $k$ or $p=i$, and $Type(NS_i^{m-1}[l_{kp}^r])=Msg\mbox{-}R$ or $Msg\mbox{-}X$ where $r\leq m-1$, and $i$ pretends $Type(State_i[l_{kp}^r])=Msg\mbox{-}O$ in round $m$. By claim \ref{claim 1}, an inconsistency is detected. Thus (\ref{equ 1}) holds.
    
    \item \textit{Case 6.} $k$ and $p \neq i$, and $Type(NS_i^{m-1}[l_{kp}^r])=Msg\mbox{-}X$ or $Msg\mbox{-}R$, and $i$ pretends $Type(State[l_{kp}^r])=Msg\mbox{-}O$ in round $m$. Since $Msg\mbox{-}R$ and $Msg\mbox{-}O$ have same meaning when computing the state of a agent, there is no benefit, which is the same as that in case 4.
    
    \item \textit{Case 7.} $k$ and $p \neq i$, and $Type(NS_i^{m-1}[l_{kp}^r])=X(r)$, and $i$ pretends $Type(State[l_{kp}^r])=X(r-1)$ in round $m$. We can turn this case into case 3, because the type of $NS_i^{m-1}[l_{kp}^{r-1}]$ must be $Msg\mbox{-}R$. So it has the same result as that in case 3. Thus, yet again, (\ref{equ 1}) holds.
    
    \item \textit{Case 8.} $k$ or $p=i$, such as $k=i$, and $Type(NS_i^{m-1}[l_{kp}^r])=Msg\mbox{-}R$ where $r\leq m-3$, and $i$ receives $Type(State_p[l_{kp}^r])=Msg\mbox{-}X$, and $i$ pretends $Type(State_p[l_{kp}^r])=Msg\mbox{-}O$ or $Msg\mbox{-}R$ in round $m$. (\romannumeral1) If $p$ is a nonfaulty agent in round $r+1$, then it must send $State_p[l_{kp}^r]$ to at least two good agents except $i$ in round $r+1$. Thus all nonfaulty agents must know $State_p[l_{kp}^r]$ in round $m$, so that $i$ does not gain. (\romannumeral2) If $p$ is faulty in round $r+1$, then since $i$ is nonfaulty in round $m-1$, $i$ is also a nonfaulty agent in round $r$ even if the link between $i$ and $p$ is faulty. Thus the results are the same as those in case 4 and case 6. Therefore (\ref{equ 1}) holds both (\romannumeral1) and (\romannumeral2).
    
    \end{itemize} 
    In summary, in either case, $i$'s utility is at least as high with $u_i(\sigma_i^{\Call{Consensus}{}}, \sigma_{-i}^{\Call{Consensus}{}})$ as with $u_i(\sigma_i, \sigma_{-i}^{\Call{Consensus}{}})$.
    
    \item \textit{Type 7.} Since the random numbers are only used in inconsistency detection, thus either $j$ detects an inconsistency and decides $\perp$ or it does not affect the final outcome if no inconsistency is detected. Thus (\ref{equ 1}) holds.
    
    \item \textit{Type 8.} It is easy to see that either an inconsistency is detected due to consensus difference or restoring secret faulty; it does not affect the outcome if $l\notin D$ or $l$ is not the agent whose preference is chosen; or $i$ does not gain due to the blind initial preferences by definition $\ref{definition 1}$ and the random agents' proposals. 
    
    \item \textit{Type 9.} Clearly it does not affect the outcome if $j$ decides $\parallel$ in the receiving phase of round $t+4$ or $j$ does not receive the messages from $i$. Otherwise, $j$ must receive the messages from at least two good agents in last round. Then $j$ detects an inconsistency and decides $\perp$. So (\ref{equ 1}) holds.
    
    \item \textit{Type 10.} We divide this type into two cases to prove.
    \begin{itemize}
    \item \textit{Case 1.} There is no consensus in the system. This case happens when either all agents become faulty due to the deviation or restoring secret faulty in round $t+3$ for all good agents because of the missing pieces of $i$.
    Thus $i$'s utility with pretending to crash is lower than with following the protocol in case 1.
    \item \textit{Case 2.} There is a consensus finally. (\romannumeral1) $m\leq m^*$. Since $i$ does not send messages to any agents before $decision \ round$, $i$ cannot exist in $decision \ set \ D$. Thus $i$'s utility also decreases when pretending the protocol. (\romannumeral2) $m>m^*$. It does not affect the outcome.
    \end{itemize}
    Therefore, (\ref{equ 1}) holds both in case 1 and 2.
    
    \end{itemize}
    Finally, concluding the proof.
\end{proof}

\section{Conclusion}\label{sec:Conclusion}
In this paper, we provide an algorithm for uniform consensus that is resilient to both omission failures and strategic manipulations. We prove our uniform consensus is a Nash equilibrium as long as $n>2t+1$, and failure patterns and initial preferences are $blind$. Additionally, we present the theory of message passing in presence of process omission failures. We argue that our research enriches the theory of fault-tolerant distributed computing, and strengthens the reliability of consensus with omission failures from the perspective of game theory. And our contribution provides a theoretical basis for the combination of distributed computing and strategic manipulations in omission failure environments, which we think is an interesting research area.

In our opinion, there are many interesting open problems and research directions. We list a few here: (a) whether an algorithm for rational uniform consensus exists if coalitions are allowed. (b) The study of rational consensus with more general types of failures, such as Byzantine failures is important. (c) With the problem setting of this paper, whether the rational consensus exists if we relax the constraint $n>2t+1$. (d) Studying the rational consensus in asynchronous system, which seems significantly more complicated. (e) Introducing the assumption of agent bounded rationality may be useful in practical scenarios.

\section*{Acknowledgment}
This work is supported by National Key R$\&$D Program of China (2018YFC0832300;2018YFC0832303)

%\begin{thebibliography}{1}

%H.~Kopka and P.~W. Daly, \emph{A Guide to \LaTeX}, 3rd~ed.\hskip 1em %plus

%\end{thebibliography}

{\small
\bibliographystyle{elsarticle-num}
\bibliography{main}
}

\end{document}